\documentclass[english,a4paper]{article}
\usepackage[T1]{fontenc}
\usepackage[latin9]{inputenc}
\usepackage[a4paper]{geometry}
\usepackage{color}
\usepackage{babel}
\usepackage{float}
\usepackage{amsthm}
\usepackage{amsmath}
\usepackage{amssymb}
\usepackage{esint}
\usepackage[authoryear]{natbib}
\usepackage[unicode=true,pdfusetitle,
 bookmarks=true,bookmarksnumbered=false,bookmarksopen=false,
 breaklinks=false,pdfborder={0 0 0},backref=false,colorlinks=true]
 {hyperref}
\hypersetup{
 citecolor=blue,urlcolor=blue}

\makeatletter

\floatstyle{ruled}
\newfloat{algorithm}{tbp}{loa}
\providecommand{\algorithmname}{Algorithm}
\floatname{algorithm}{\protect\algorithmname}

\newcommand{\lyxaddress}[1]{
\par {\raggedright #1
\vspace{1.4em}
\noindent\par}
}
 \theoremstyle{definition}
 \newtheorem*{defn*}{\protect\definitionname}
  \theoremstyle{definition}
  \newtheorem{problem}{\protect\problemname}
  \theoremstyle{plain}
  \newtheorem{prop}{\protect\propositionname}
  \theoremstyle{plain}
  \newtheorem{cor}{\protect\corollaryname}
  \theoremstyle{remark}
  \newtheorem*{rem*}{\protect\remarkname}
  \theoremstyle{plain}
  \newtheorem{lem}{\protect\lemmaname}

\makeatother

  \providecommand{\definitionname}{Definition}
  \providecommand{\lemmaname}{Lemma}
  \providecommand{\problemname}{Problem}
  \providecommand{\propositionname}{Proposition}
  \providecommand{\remarkname}{Remark}
\providecommand{\corollaryname}{Corollary}

\begin{document}

\title{Perfect simulation using atomic regeneration with application to
Sequential Monte Carlo}

\author{Anthony Lee$^{*}$, Arnaud Doucet$^{\dagger}$ and Krzysztof {\L}atuszy{\'n}ski$^{*}$}

\maketitle

\lyxaddress{$^{*}$Department of Statistics, University of Warwick, UK }

\lyxaddress{$^{\dagger}$Department of Statistics, University of Oxford, UK.}
\begin{abstract}
Consider an irreducible, Harris recurrent Markov chain of transition
kernel $\Pi$ and invariant probability measure $\pi$. If $\Pi$
satisfies a minorization condition, then the split chain allows the
identification of regeneration times which may be exploited to obtain
perfect samples from $\pi$. Unfortunately, many transition kernels
associated with complex Markov chain Monte Carlo algorithms are analytically
intractable, so establishing a minorization condition and simulating
the split chain is challenging, if not impossible. For uniformly ergodic
Markov chains with intractable transition kernels, we propose two
efficient perfect simulation procedures of similar expected running
time which are instances of the multigamma coupler and an imputation
scheme. These algorithms overcome the intractability of the kernel
by introducing an artificial atom and using a Bernoulli factory. We
detail an application of these procedures when $\Pi$ is the recently
introduced iterated conditional Sequential Monte Carlo kernel. We
additionally provide results on the general applicability of the methodology,
and how Sequential Monte Carlo methods may be used to facilitate perfect
simulation and/or unbiased estimation of expectations with respect
to the stationary distribution of a non-uniformly ergodic Markov chain.
\end{abstract}
{\small{}Keywords: Artificial atom; Bernoulli factory; Markov chain
Monte Carlo; Perfect simulation; Regeneration; Sequential Monte Carlo.}{\small \par}

\section{Introduction\label{sec:Intro}}

Given a target probability measure $\pi$ on a general state space
$\left(\mathsf{X},\mathcal{B}(\mathsf{X})\right)$, the key idea of
Markov chain Monte Carlo (MCMC) is to simulate a $\pi$-irreducible,
Harris recurrent Markov chain $\mathbf{X}:=(X_{n})_{n\geq1}$ with
$\pi$-invariant transition kernel $\Pi:\mathsf{X}\times\mathcal{B}(\mathsf{X})\rightarrow[0,1]$
to generate asymptotically samples from $\pi$. In practice, the Markov
chain is only simulated for a finite number of iterations and so we
do not obtain samples exactly distributed according to $\pi$. To
address this problem, much effort has been put into the development
of perfect simulation methods over the past twenty years. In particular,
the Coupling From The Past (CFTP) procedure of \citet{propp1996exact}
has enabled the development of perfect simulation algorithms for the
Ising model and various spatial point processes models \citep[see, e.g.,][]{mollerbook2004}.
The applications of CFTP to general state spaces remain limited as,
for implementation purposes, the MCMC kernel is typically required
to satisfy strong stochastic monotonicity properties \citep[see, e.g.,][]{propp1996exact,foss1998perfect}. 

We focus here on an alternative set of techniques based on regeneration.
Henceforth, the first assumption we will make is that $\Pi$ satisfies
a one-step minorization condition; that is we have for all $x\in\mathsf{X}$
\begin{equation}
\Pi(x,{\rm d}y)\geq s(x)\nu({\rm d}y),\label{eq:minorization_general}
\end{equation}
where $s:\mathsf{X}\rightarrow[0,1]$, $\nu$ is a given (regeneration)
probability measure and $\pi\left(s\right):=\int_{\mathsf{X}}s(x)\pi\left({\rm d}x\right)>0$.
In this context, the introduction of an associated ``split chain''
\citep{nummelin1978splitting,athreya1978new} allows the identification
of regeneration times and provides a mixture representation of $\pi$
\citep{asmussen1992stationarity,hobert2004mixture}.

This mixture representation was exploited by \citet{hobert2004mixture}
to obtain approximate samples from $\pi$, whereas \citet{blanchet2007exact}
and \citet{flegal2012exact} used it to obtain exact simulation algorithms
but their procedures have an infinite expected running time. In the
more restrictive scenario where one can take $s=\epsilon$ to be a
constant function with $\epsilon>0$, implying that $\Pi$ is uniformly
ergodic, the same mixture representation can be used to derive the
multigamma coupler of \citet{murdoch1998exact}, an exact simulation
algorithm with finite expected running time. Unfortunately, all of
these techniques are fairly restrictive in the sense that they require
being able to simulate the split chain, which can be difficult for
sophisticated MCMC kernels. For example, the iterated conditional
Sequential Monte Carlo (i-cSMC) kernel is an MCMC kernel introduced
in \citet{Andrieu2010} to sample from high-dimensional target distributions.
It has been established recently by \citet{chopin:singh:2013}, \citet{andrieu2013uniform}
and \citet{lindsten_pg} that this kernel satisfies the minorization
condition (\ref{eq:minorization_general}) with $s=\epsilon>0$ where
$\epsilon$ can be known explicitly, yet one is unable to simulate
the split chain associated to the i-cSMC kernel as this kernel does
not admit a tractable expression.

The main contribution of this paper is to develop two general-purpose
procedures for perfect simulation when the Markov transition kernel
$\Pi$ admits a singleton atom $\alpha=\{a\}$, for some $a\in\mathsf{X}$,
for which $\inf_{x\in\mathsf{X}}\Pi(x,\alpha)\geq\beta>0$ with $\beta$
known. This, together with the ability to simulate the Markov chain
$\mathbf{X}$, is the only requirement for perfect simulation to be
implemented. While this assumes that $\mathbf{X}$ is uniformly ergodic,
and knowledge of $\beta$ is non-trivial in general, it is a significant
relaxation of the conditions needed by other perfect simulation algorithms
on general state spaces. The key mechanism that allows the identification
of perfect samples is the use of a Bernoulli factory \citep{keane1994bernoulli}.
We show that for specific implementations, the expected number of
Markov chain transitions required to obtain a perfect sample is in
$\mathcal{O}(\beta^{-1})$. Since many Markov transition kernels used
in statistical applications do not admit a singleton atom, we overview
relevant strategies in \citet{brockwell2005identification} to modify
$\Pi$ to introduce an artificial singleton atom and obtain an associated
invariant distribution which is a mixture of $\pi$ and a point mass
at this artificial atom. This modification can be carried out in a
way which ensures that uniform ergodicity of the original Markov kernel
is inherited by the modified Markov kernel.

While our methodology is generally applicable, we primarily focus
in this paper on perfect simulation from the path distribution of
a discrete-time Feynman--Kac model using an i-cSMC kernel. While it
can be established that standard SMC methods provide samples whose
distribution can be made arbitrary close to the path distribution
of interest by increasing the number of particles, no perfect simulation
method has hitherto been devised with guarantees of expected time
polynomial in the time horizon, $n$. Under regularity assumptions
on the Feynman--Kac model, we show that our methodology requires expected
$\mathcal{O}(n^{2})$ time to generate perfect samples, and may be
implementable in expected polynomial time under less restrictive assumptions.
The introduction of the artificial atom in this case follows from
a simple and generally applicable extension of the Feynman--Kac model
of interest. This allows us to sample from the joint posterior distribution
of latent variables in a hidden Markov model (HMM), as long as the
chosen parameter $\beta$ satisfies $\inf_{x\in\mathsf{X}}\Pi(x,\alpha)\geq\beta$
. In practice, we have found that with appropriate algorithm settings
one can take $\beta$ to be fairly large with no indication that this
assumption fails to hold. Additionally a diagnostic can be performed
to check this assumption during the course of the algorithm.

The rest of the paper is organized as follows. In Section~\ref{sec:Reg-and-PS}
we review the split chain construction, its simulation and how this
chain may be exploited to obtain perfect samples from $\pi$. In Section~\ref{sec:Reg-PS-perfectatom}
we consider the case where $\Pi$ admits a singleton atom, and show
how perfect simulation methods may be viably implemented through the
use of an appropriate Bernoulli factory. In Section~\ref{sec:artificial_atom}
we overview the general principles behind modification of $\Pi$ to
create a Markov kernel that admits an artificial singleton atom. Section~\ref{sec:icsmc_atom}
constitutes the major application, in which we combine results from
Section~\ref{sec:Reg-PS-perfectatom} with a novel variant of an
approach in Section~\ref{sec:artificial_atom} to develop perfect
simulation methodology for sampling from a Feynman--Kac law on the
path space. In Section~\ref{sec:discussion} we discuss possible
approaches for dealing with the case where the constant $\beta$ is
not known explicitly and show how the method we have introduced may
be useful as a component of a more general perfect simulation algorithm
or unbiased estimation scheme. We also reiterate the potential utility
of parallel implementation of regenerative Markov chains. Section~\ref{sec:Applications}
demonstrates our methodology on a number of applications.

\section{Regeneration and perfect simulation\label{sec:Reg-and-PS}}

The following notation will be used throughout the paper. For a metric
space $\mathsf{X}$, we denote by $\mathcal{B}(\mathsf{X})$ the Borel
$\sigma$-algebra on $\mathsf{X}$. For $\mu:\mathcal{B}(\mathsf{X})\rightarrow\mathbb{R}_{+}$
a measure, $P:\mathsf{X}\times\mathcal{B}(\mathsf{X})\rightarrow[0,1]$
a Markov kernel and $f$ a real-valued measurable function w.r.t.
$\mathcal{B}(\mathsf{X})$, we write $\mu(f):=\int_{\mathsf{X}}f(x)\mu({\rm d}x)$
and define $\mu P$ to be the measure satisfying $\mu P({\rm d}y):=\int_{\mathsf{X}}P(x,{\rm d}y)\mu({\rm d}x)$.
We set $P^{1}:=P$ and $P^{n}\left(x,{\rm d}y\right):=\int_{\mathsf{X}}P^{n-1}(x,{\rm d}z)P(z,{\rm d}y)$
for $n\geq2$. When a measure $\mu$ admits a density w.r.t.\ some
dominating measure, we will denote both the density and the measure
by $\mu$, so e.g., $\mu(x)$ is the density at $x\in\mathsf{X}$
and $\mu(A)$ the measure of the set $A\in\mathcal{B}(\mathsf{X})$.
When $\mu$ is additionally a probability measure, we will also refer
to it as a distribution. If $f:\mathsf{X}\rightarrow\mathbb{R}$ is
a function and $c\in\mathbb{R}$ a constant, we will write $f=c$
to mean $f$ is the constant function with $x\mapsto c$ for all $x\in\mathsf{X}$.
Finally, $\delta_{x}$ denotes the Dirac measure centred at $x$.

\subsection{Atoms, split chain and regeneration}

All forthcoming developments rely on the notion of an atom (see, e.g.,
\citealt[Definition~4.3]{nummelin1984general} or \citealt[Chapter~5]{meyn2009markov}),
which we now introduce. Consider an irreducible, Harris recurrent
Markov chain $\mathbf{Y}:=(Y_{n})_{n\geq1}$ of transition kernel
$P$ of invariant distribution $\chi$ on a measurable space $(E,\mathcal{E})$.
A set $\alpha\in\mathcal{E}$ is a \emph{proper atom} for $P$ if
there exists a probability measure $\mu:\mathcal{E}\rightarrow[0,1]$
such that
\[
P(y,A)=\mu(A),\quad y\in\alpha,A\in\mathcal{E},
\]
that is each time the chain enters the atom, its next state is sampled
according to the \emph{regeneration measur}e $\mu$. Since $\mathbf{Y}$
is Harris recurrent and $\chi$-irreducible, such an atom $\alpha$
is \emph{accessible} if $\chi(\alpha)>0$ and $\mathbf{Y}$ returns
infinitely often to $\alpha$ with probability $1$.

In general state spaces, it is rare for a Markov kernel to admit a
proper, accessible atom. The major contribution of \citet{nummelin1978splitting}
and \citet{athreya1978new} was to show that, even if the Markov chain
$\mathbf{X}$ with transition kernel $\Pi$ defined in Section~\ref{sec:Intro}
does not admit a proper atom, one can exploit the minorization (\ref{eq:minorization_general})
to construct a bivariate Markov chain $\tilde{\mathbf{X}}_{\nu,s}$
evolving on the extended space $\mathsf{X}\times\{0,1\}$ which admits
a proper, accessible atom. We provide here a brief summary of the
construction of \citet{nummelin1978splitting}.

Using (\ref{eq:minorization_general}), we can write $\Pi(x,{\rm d}y)$
as a mixture
\begin{equation}
\Pi(x,{\rm d}y)=s(x)\nu({\rm d}y)+[1-s(x)]R_{\nu,s}(x,{\rm d}y),\label{eq:decompose_Pi_s}
\end{equation}
where $R_{\nu,s}$ is a \emph{residual kernel}, defined for $x\in\mathsf{X}$
with $s(x)<1$ by 
\begin{equation}
R_{\nu,s}(x,{\rm d}y):=\frac{\Pi(x,{\rm d}y)-s(x)\nu({\rm d}y)}{1-s(x)}.\label{eq:residual_kernel_general}
\end{equation}
The subscript $(\nu,s)$ in $R_{\nu,s}$ emphasizes that there are
many choices of $\nu$ and $s$ in (\ref{eq:minorization_general})
for a given $\Pi$, each associated to a particular residual kernel.
The split chain is then defined as follows.
\begin{defn*}
The \emph{split chain} is a bivariate Markov chain $\tilde{\mathbf{X}}_{\nu,s}=(\tilde{X}_{n}^{(\nu,s)})_{n\geq1}=(Z_{n}^{(\nu,s)},\rho_{n}^{(\nu,s)})_{n\geq1}$
on the extended space $\mathsf{X}\times\{0,1\}$ with transition kernel
\begin{equation}
\tilde{\Pi}_{\nu,s}(x,\rho;{\rm d}y,\varrho):=\left\{ \mathbb{I}(\rho=1)\nu({\rm d}y)+\mathbb{I}(\rho=0)R_{\nu,s}(x,{\rm d}y)\right\} s(y)^{\varrho}\left[1-s(y)\right]^{1-\varrho},\label{eq:split_chain_transition_kernel}
\end{equation}
and invariant distribution $\tilde{\pi}_{\nu,s}({\rm d}x,\rho):=\pi\left({\rm d}x\right)s(x)^{\rho}\left[1-s(x)\right]^{1-\rho}$.
\end{defn*}
This construction, combined with the fact that $\mathbf{X}$ is Harris
recurrent and $\pi$-irreducible, implies key properties of $\tilde{\mathbf{X}}_{\nu,s}$:
\begin{enumerate}
\item If the law of $Z_{1}^{(\nu,s)}$ is equal to the law of $X_{1}$ then
the laws of $(Z_{n}^{(\nu,s)})_{n\geq1}$ and $(X_{n})_{n\geq1}$
are identical \citep[Theorem~1]{nummelin1978splitting}.
\item $\tilde{\mathbf{X}}_{\nu,s}$ is Harris recurrent with $\mathsf{X}\times\{1\}$
a proper, accessible atom for $\tilde{\Pi}_{\nu,s}$ with associated
\emph{regeneration measure} $\tilde{\nu}_{\nu,s}({\rm d}x,\rho):=\nu({\rm d}x)s(x)^{\rho}\left[1-s(x)\right]^{1-\rho}$
\citep[Theorem~2]{nummelin1978splitting}.
\end{enumerate}
Crucially, one can see that $\tilde{\Pi}_{\nu,s}$ admits $\mathsf{X}\times\{1\}$
as a proper, accessible atom even when the original Markov transition
kernel $\Pi$ admits no proper, accessible atom. To simplify notation
and emphasize that only the second coordinate of $\tilde{\mathbf{X}}_{\nu,s}=(Z_{n}^{(\nu,s)},\rho_{n}^{(\nu,s)})_{n\geq1}$
is dependent on $\nu$ and $s$ in (\ref{eq:minorization_general}),
we dispense with $(Z_{n}^{(\nu,s)})_{n\geq1}$ and hereafter define
$\tilde{\mathbf{X}}_{\nu,s}:=(\tilde{X}_{n}^{(\nu,s)})_{n\geq1}:=(X_{n},\rho_{n}^{(\nu,s)})_{n\geq1}$,
the Markov chain with transition kernel $\tilde{\Pi}_{\nu,s}$. In
this split chain, the variables $(\rho_{n}^{(\nu,s)})_{n\geq1}$ are
indicators of regeneration since $\tilde{X}_{n}\mid\{\rho_{n-1}^{(\nu,s)}=1\}\sim\tilde{\nu}_{\nu,s}$.

When $\rho_{n-1}^{(\nu,s)}=1$, it is customary to call $\tilde{X}_{n-1}^{(\nu,s)}$
the sample just prior to regeneration, since it is $\tilde{X}_{n}^{(\nu,s)}$
that is distributed according to $\tilde{\nu}_{\nu,s}$. Due to the
close relationship between $\mathbf{X}$ and $\tilde{\mathbf{X}}_{\nu,s}$,
$\nu$ is commonly also referred to as a regeneration measure for
$\mathbf{X}$, the construction of an associated split chain $\tilde{\mathbf{X}}_{\nu,s}$
being implicit.

We define the sequence of regeneration times $(\tau_{k}^{(\nu,s)})_{k\geq1}$
of the Markov chain $\tilde{\mathbf{X}}_{\nu,s}$ via
\[
\tau_{k}^{(\nu,s)}:=\min\{n>\tau_{k-1}^{(\nu,s)}:\rho_{n}^{(\nu,s)}=1\},
\]
with $\tau_{0}^{(\nu,s)}:=0$. In the sequel, $\tau_{\nu,s}:=\tau_{1}^{(\nu,s)}$
will be used to denote the first regeneration time where no ambiguity
can result. The introduction of this split chain was a major innovation
in the analysis of Markov chains on general state spaces, because
it allows the Markov chain $\tilde{\mathbf{X}}_{\nu,s}$ to be partitioned
into i.i.d.\ tours, where tour $i$, $i\in\mathbb{N}$, is defined
as $(\tilde{X}_{n}^{(\nu,s)}:\tau_{i-1}^{(\nu,s)}<n\leq\tau_{i}^{(\nu,s)})$.

\subsection{Simulating the split chain\label{sub:mykland_simulation}}

Although the split chain $\tilde{\mathbf{X}}_{\nu,s}$ was originally
introduced as a theoretical tool to analyze the Markov chain $\mathbf{X}$,
statistical methodology has since been developed which requires simulating
$\tilde{\mathbf{X}}_{\nu,s}$ \citep[see, e.g., ][]{mykland1995regeneration,hobert2002applicability}.
In practice, however, sampling from the split chain kernel $\tilde{\Pi}$
is not always feasible even when sampling from $\Pi$ is.

A standard approach due to \citet{mykland1995regeneration} consists
of simulating $\mathbf{X}$ and then imputing the values of $(\rho_{n}^{(\nu,s)})_{n\geq1}$
conditional upon $\mathbf{X}$. Letting $\mathbb{P}_{\nu,s}$ denote
the law of $\tilde{\mathbf{X}}_{\nu,s}$ when $\tilde{X}_{1}\sim\tilde{\nu}_{\nu,s}$,
they observe that $\rho_{n-1}^{(\nu,s)}$ depends on $\tilde{\mathbf{X}}_{\nu,s}$
only through $X_{n-1}$ and $X_{n}$ and that
\begin{equation}
\mathbb{P}_{\nu,s}\left(\rho_{n-1}^{(\nu,s)}=1\mid X_{n-1}=x_{n-1},X_{n}=x_{n}\right)=\frac{s(x_{n-1})\nu({\rm d}x_{n})}{\Pi(x_{n-1},{\rm d}x_{n})},\label{eq:mykland_trick}
\end{equation}
which can be computed in a variety of situations, but not in general.
For example, in many cases $\Pi$ can be expressed as
\[
\Pi(x,{\rm d}y)=Q_{x}({\rm d}y)a(x,y)+r(x)\delta_{x}({\rm d}y),
\]
where $a:\mathsf{X}\times\mathsf{Y}\rightarrow[0,1]$ and $r:\mathsf{X}\rightarrow[0,1]$,
and $\{Q_{x}:x\in\mathsf{X}\}$ and $\nu$ admit densities with respect
to a common dominating measure without point masses. Then $\mathbb{P}_{\nu,s}(\rho_{n-1}^{(\nu,s)}=1\mid X_{n-1}=X_{n})=0$
and for any $X_{n-1}\neq X_{n}$ the r.h.s.\ of (\ref{eq:mykland_trick})
is given by \citep[Proposition~1]{mykland1995regeneration}
\[
\frac{s(x_{n-1})\nu(x_{n})}{Q_{x_{n-1}}(x_{n})a(x_{n-1},x_{n})}.
\]
It is apparent, however, that imputing $\rho_{n-1}^{(\nu,s)}$ in
practice may be impossible when this expression cannot be computed
for distinct $x_{n-1},x_{n}\in\mathsf{X}$. We will show in Section~\ref{sec:Reg-PS-perfectatom}
how we can overcome this limitation when $\Pi$ admits a singleton
atom.

\subsection{Mixture representation of the invariant measure and perfect simulation\label{sub:mix_rep_ps}}

It was established in \citet{asmussen1992stationarity}, \citet{hobert2004mixture}
and \citet{hobert2006using} that the split chain construction provides
the following mixture representation of the target distribution $\pi$:
\begin{equation}
\pi({\rm d}x)=\sum_{n=1}^{\infty}\frac{\mathbb{P}_{\nu,s}(\tau_{\nu,s}\geq n)}{\mathbb{E}_{\nu,s}(\tau_{\nu,s})}\eta_{n}^{(\nu,s)}({\rm d}x),\label{eq:mixture_representation}
\end{equation}
where, defining $G_{s}(x):=1-s(x)$ and with the convention that $\prod_{\emptyset}=1$,
\begin{eqnarray}
\eta_{n}^{(\nu,s)}(A) & = & \mathbb{P}_{\nu,s}\left(X_{n}\in A\mid\tau_{\nu,s}\geq n\right)\nonumber \\
 & = & \frac{\int_{\mathsf{X}^{n}}\mathbb{I}\left\{ x_{n}\in A\right\} \left[\prod_{p=2}^{n}G_{s}(x_{p-1})\right]\nu({\rm d}x_{1})\prod_{p=2}^{n}R_{\nu,s}(x_{p-1},{\rm d}x_{p})}{\int_{\mathsf{X}^{n}}\left[\prod_{p=2}^{n}G_{s}(x_{p-1})\right]\nu({\rm d}x_{1})\prod_{p=2}^{n}R_{\nu,s}(x_{p-1},{\rm d}x_{p})}.\label{eq:eta_n}
\end{eqnarray}
This mixture representation (\ref{eq:mixture_representation}) implies
that if one were able to sample a $\mathbb{N}$-valued random variable
$N$ of probability mass function (p.m.f.) 
\begin{equation}
\Pr(N=n)=\frac{\mathbb{P}_{\nu,s}(\tau_{\nu,s}\geq n)}{\mathbb{E}_{\nu,s}(\tau_{\nu,s})}\label{eq:Ngeneraldist}
\end{equation}
and $\xi|\left(N=n\right)\sim\eta_{n}^{(\nu,s)}$ then unconditionally
$\xi\sim\pi$. Algorithm~\ref{alg:psim_general_regen} describes
this procedure.

\begin{algorithm}[H]
\protect\caption{A generic regenerative perfect simulation algorithm\label{alg:psim_general_regen}}

\begin{enumerate}
\item Sample $N$ from the distribution with p.m.f. (\ref{eq:Ngeneraldist}).
\item Output a sample from $\eta_{N}^{(\nu,s)}$.\end{enumerate}
\end{algorithm}

It has been realized in \citet{blanchet2007exact} and \citet{flegal2012exact}
that a Bernoulli factory can be used in some circumstances to facilitate
sampling $N$ according to (\ref{eq:Ngeneraldist}), although we note
that the recent literature on the Bernoulli factory starting with
\citet{keane1994bernoulli} was itself motivated by the regenerative
simulation work of \citet{asmussen1992stationarity}. For practical
implementation, this typically requires accurate upper bounds on $\mathbb{P}_{\nu,s}(\tau_{\nu,s}\geq n)$
as a function of $n$, and \citet{flegal2012exact} provide examples
where this is possible. Sampling from $\eta_{N}^{(\nu,s)}$, and thereby
obtaining a sample from $\pi$, is then achieved by using a simple
rejection algorithm where the proposal is $\nu R_{\nu,s}^{N-1}$.
Unfortunately, this rejection technique has an infinite expected running
time when $N$ is distributed according to (\ref{eq:Ngeneraldist})
\citep[Proposition~2]{asmussen1992stationarity,blanchet2007exact}.

As noted in \citet{hobert2004mixture}, the problem becomes much simpler
when one considers the case where (\ref{eq:minorization_general})
holds with $s=\epsilon$. The representation (\ref{eq:mixture_representation})
adapted to the minorization (\ref{eq:minorization_general}) with
$s=\epsilon$ yields 
\begin{equation}
\pi({\rm d}x)=\sum_{n=1}^{\infty}\epsilon(1-\epsilon)^{n-1}\nu R_{\nu,\epsilon}^{n-1}({\rm d}x).\label{eq:mixture_uniform_minorization}
\end{equation}
Indeed, $\tau_{\nu,\epsilon}$ is a geometric random variable with
success probability $\epsilon$ so 
\[
\frac{\mathbb{P}_{\nu,\epsilon}(\tau_{\nu,\epsilon}\geq n)}{\mathbb{E}_{\nu,\epsilon}(\tau_{\nu,\epsilon})}=\epsilon(1-\epsilon)^{n-1}=\mathbb{P}_{\nu,\epsilon}(\tau_{\nu,\epsilon}=n),
\]
and $\eta_{n}^{(\nu,\epsilon)}=\nu R_{\nu,\epsilon}^{n-1}$ because
$s(x)=\epsilon$ is independent of $x$. There are two perfect simulation
procedures that arise from alternative interpretations of (\ref{eq:mixture_uniform_minorization}),
which we now outline.

The first interpretation is that if $\tilde{X}_{1}\sim\tilde{\nu}_{\nu,\epsilon}$
then $X_{\tau_{\nu,\epsilon}}$, the first coordinate of $\tilde{\mathbf{X}}_{\nu,\epsilon}$
at time $\tau_{\nu,\epsilon}$, is a perfect sample from $\pi$. In
other words, the sample $X_{\tau_{\nu,\epsilon}}$ just prior to regeneration
is a perfect sample from $\pi$, and Algorithm~\ref{alg:psim_onestep_generic_split}
is a corresponding perfect simulation procedure using (\ref{eq:mykland_trick})
to determine when the first regeneration has occurred. The algorithm
is practical whenever $\epsilon$ is known and one can compute the
Radon--Nikodym derivative appearing in (\ref{eq:mykland_psim_onestep_generic}).

\begin{algorithm}[H]
\protect\caption{Regenerative perfect simulation algorithm when $s=\epsilon$ via imputation\label{alg:psim_onestep_generic_split}}

\begin{enumerate}
\item Sample $X_{1}\sim\nu$.
\item For $n=2,3,\ldots$

\begin{enumerate}
\item Sample $X_{n}\sim\Pi(X_{n-1},\cdot)$.
\item With probability
\begin{equation}
\epsilon\frac{{\rm d}\nu(\cdot)}{{\rm d}\Pi(X_{n-1},\cdot)}(X_{n}),\label{eq:mykland_psim_onestep_generic}
\end{equation}
stop and output $X_{n-1}$.\end{enumerate}
\end{enumerate}
\end{algorithm}

A second interpretation is that if one samples $N\sim{\rm Geometric}(\epsilon)$
followed by $\xi\sim\nu R_{\nu,\epsilon}^{N-1}$, then unconditionally
$\xi\sim\pi$, and Algorithm~\ref{alg:psim_onestep_generic} is a
procedure based upon this observation. The algorithm is practical
as soon as $\epsilon$ is known and both $\nu$ and $R_{\nu,\epsilon}$
can be sampled from. This latter scheme corresponds to the multigamma
coupler of \citet{murdoch1998exact} but without any explicit appeal
to CFTP as noticed by \citet{hobert2004mixture}.

\begin{algorithm}[H]
\protect\caption{Regenerative perfect simulation algorithm when $s=\epsilon$ via the
multigamma coupler\label{alg:psim_onestep_generic}}

\begin{enumerate}
\item Sample $N\sim{\rm Geometric}(\epsilon)$.
\item Output a sample from $\nu R_{\nu,\epsilon}^{N-1}$.\end{enumerate}
\end{algorithm}

While cast as perfect simulation algorithms, both of these procedures
essentially simulate a single tour of the split chain $\tilde{\mathbf{X}}_{\nu,\epsilon}$.
In practice, it is often the case that one cannot sample from $\nu$
or $R_{\nu,\epsilon}$, or compute the Radon-Nikodym derivative (\ref{eq:mykland_psim_onestep_generic})
so that Algorithms~\ref{alg:psim_onestep_generic_split} and~\ref{alg:psim_onestep_generic}
cannot be implemented. We now discuss a special case where such an
implementation is feasible.

\section{Regeneration and perfect simulation for Markov chains with a singleton
atom\label{sec:Reg-PS-perfectatom}}

\subsection{Markov chains with singleton atoms\label{sub:MC_singleton_atoms}}

Assume that the Markov chain $\mathbf{X}$ admits not only a proper,
accessible atom $\alpha$, but one which is additionally a \emph{singleton},
i.e., $\alpha=\{a\}$ for some distinguished $a\in\mathsf{X}$ assumed
known to the user. This special case is a central focus of this paper.
Most Markov kernels do not admit such a singleton atom but we will
show in Section~\ref{sec:artificial_atom} how such an atom can be
introduced fairly generally, following \citet{brockwell2005identification}.
When this assumption is met, $\mathbf{X}$ will visit the point $a$
infinitely often.

A natural split chain to consider in this context is obtained by taking
$s=p$ and $\nu=\delta_{a}$ in (\ref{eq:minorization_general}) where
\[
p(x):=\Pi(x,\{a\}).
\]
To simplify notation, we write $\tilde{\mathbf{X}}_{a,p}$, $\rho_{n}^{(a,p)}$
and $\mathbb{P}_{a,p}$ for $\tilde{\mathbf{X}}_{\delta_{a},p}$,
$\rho_{n}^{(\delta_{a},p)}$ and $\mathbb{P}_{\delta_{a},p}$ respectively.
The split chain $\tilde{\mathbf{X}}_{a,p}$ can be very easily simulated
in this case using the imputation method discussed in Section~\ref{sub:mykland_simulation}
as (\ref{eq:mykland_trick}) reduces to 
\[
\mathbb{P}_{a,p}\left(\rho_{n-1}^{(a,p)}=1\mid X_{n-1}=x_{n-1},X_{n}=x_{n}\right)=\mathbb{I}\left(X_{n}=a\right).
\]

However, performing perfect simulation using Algorithm~\ref{alg:psim_general_regen}
in this scenario remains challenging as all currently available implementations
have an infinite expected running time \citep{asmussen1992stationarity,blanchet2007exact,flegal2012exact}.
We show how to bypass this problem when $p(x)$ satisfies an additional
assumption.

\subsection{Perfect simulation using the split chain $\tilde{\mathbf{X}}_{a,\epsilon}$\label{sub:indirect_s_e}}

We develop here practical perfect simulation schemes under the additional
assumption that for some known constant $\beta>0$  
\[
\underline{p}:=\inf_{x\in\mathsf{X}}p(x)\geq\beta.
\]
This assumption implies that $\mathbf{X}$ is uniformly ergodic and
satisfies the one-step minorization condition $\Pi(x,\{a\})\geq\beta$.
The split chain $\tilde{\mathbf{X}}_{a,\epsilon}$ obtained by taking
$s=\epsilon$ and $\nu=\delta_{a}$ in (\ref{eq:minorization_general})
for some $\epsilon\in(0,\beta)$ is at the heart of our methodology.
Contrary to the natural split chain $\tilde{\mathbf{X}}_{a,p}$ introduced
in Section~\ref{sub:MC_singleton_atoms} which does not lead to a
currently implementable perfect simulation algorithm enjoying a finite
expected running time, the introduction of the alternative split chain
$\tilde{\mathbf{X}}_{a,\epsilon}$ is motivated by Section~\ref{sub:mix_rep_ps}
where it is shown that perfect simulation from $\pi$ is feasible
if either Algorithm~\ref{alg:psim_onestep_generic_split} or Algorithm~\ref{alg:psim_onestep_generic}
can be implemented.

We show here how one can implement both Algorithms~\ref{alg:psim_onestep_generic_split}
and~\ref{alg:psim_onestep_generic} using an appropriate Bernoulli
factory. Our algorithms exploit the fact that although the expression
of $p(x)$ is usually unknown, one can simulate a Bernoulli random
variable with success probability $p(x)$, for any required $x\in\mathsf{X}$,
by sampling $Y\sim\Pi(x,\cdot)$ and outputting $\mathbb{I}\{Y=a\}$.
We will generically refer to a $p$-coin as a coin whose flips are
independent Bernoulli random variables with success probability $p$.
Details of the solutions to the Bernoulli factory problems are presented
in Section~\ref{sub:Bernoulli-factory-algorithms}.

Algorithm~\ref{alg:psim_onestep_generic_split} involves two operations:
simulation from $\Pi$ and simulation of a Bernoulli random variable
with success probability (\ref{eq:mykland_psim_onestep_generic}).
The former is possible by assumption and the latter reduces in this
case to
\begin{equation}
\mathbb{P}_{a,\epsilon}\left(\rho_{n-1}^{(a,\epsilon)}=1\mid X_{n-1}=x_{n-1},X_{n}=x_{n}\right)=\frac{\epsilon}{p(x_{n-1})}\mathbb{I}\{x_{n}=a\},\label{eq:e_over_p_coin}
\end{equation}
which is possible as we have a solution to the following Bernoulli
factory problem.
\begin{problem}
\label{prob:bf1}Given flips of a $p$-coin with $p\geq\beta>\epsilon>0$,
for known constants $\epsilon$ and $\beta$, simulate an $\frac{\epsilon}{p}$-coin
flip.
\end{problem}
The practical implementation of Algorithm~\ref{alg:psim_onestep_generic_split}
is summarized in Algorithm~\ref{alg:Practical_imputation}.

\begin{algorithm}[H]
\protect\caption{Practical implementation of Algorithm~\ref{alg:psim_onestep_generic_split}\label{alg:Practical_imputation}}

\begin{enumerate}
\item Choose $\beta\in(0,\underline{p}]$ and $\epsilon\in(0,\beta)$. We
suggest $\epsilon=\beta/2$.
\item Set $X_{1}=a$.
\item For $n=2,3,\ldots$:

\begin{enumerate}
\item Sample $X_{n}\sim\Pi(X_{n-1},\cdot)$.
\item If $X_{n}=a$, sample $\rho_{n-1}^{(a,\epsilon)}\sim{\rm Bernoulli}(\epsilon/p(X_{n-1}))$.
Otherwise set $\rho_{n-1}^{(a,\epsilon)}=0$.\label{enu:eoverp_coin_in_algo}
\item If $\rho_{n-1}^{(a,\epsilon)}=1$, stop and output $X_{n-1}$.\end{enumerate}
\end{enumerate}
\end{algorithm}

The main computational task in Algorithm~\ref{alg:psim_onestep_generic}
is simulation from the residual kernel $R_{a,\epsilon}$, which can
be decomposed into a mixture of $R_{a,p}$ and $\nu=\delta_{a}$
\begin{eqnarray}
R_{a,\epsilon}(x,{\rm d}y) & = & \frac{\Pi(x,{\rm d}y)-\epsilon\delta_{a}({\rm d}y)}{1-\epsilon}\nonumber \\
 & = & \frac{1-p(x)}{1-\epsilon}R_{a,p}(x,{\rm d}y)+\frac{p(x)-\epsilon}{1-\epsilon}\delta_{a}({\rm d}y),\label{eq:mixture_residual_epsilon}
\end{eqnarray}
where we can simulate from $R_{a,p}(x,\cdot)$ by a rejection method
which simulates from the proposal $\Pi(x,\cdot)$ and outputs the
first sample distinct from $a$. The difficulty in simulating from
$R_{a,\epsilon}(x,\cdot)$ is in flipping a $\{1-p(x)\}/(1-\epsilon)$-coin
to select the mixture component. However, this is possible as we have
a solution to the following second Bernoulli factory problem.
\begin{problem}
\label{prob:bf2}Given flips of a $p$-coin with $p\geq\beta>\epsilon>0$,
for known constants $\epsilon$ and $\beta$, simulate a $(1-p)/(1-\epsilon)$-coin
flip.
\end{problem}
The practical implementation of Algorithm~\ref{alg:psim_onestep_generic}
is summarized in Algorithm~\ref{alg:Practical_multigamma}.

\begin{algorithm}[H]
\protect\caption{Practical implementation of Algorithm~\ref{alg:psim_onestep_generic}\label{alg:Practical_multigamma}}

\begin{enumerate}
\item Choose $\beta\in(0,\underline{p}]$ and $\epsilon\in(0,\beta)$. We
suggest $\epsilon=\beta/2$.
\item Sample $N\sim{\rm Geometric}(\epsilon)$.
\item Set $X_{1}=a$.
\item For $n=2,\ldots,N$:

\begin{enumerate}
\item Sample $Y_{n}\sim{\rm Bernoulli}\left(\{1-p(X_{n-1})\}/\{1-\epsilon\}\right)$.
\item If $Y_{n}=1$, sample $X_{n}\sim R_{a,p}(X_{n-1},\cdot)$ by rejection.
Otherwise, set $X_{n}=a$.
\end{enumerate}
\item Output $X_{N}$.\end{enumerate}
\end{algorithm}

\subsection{Bernoulli factory algorithms\label{sub:Bernoulli-factory-algorithms}}

A Bernoulli factory for a known function $f$ is an algorithm for
simulating a flip of a $f(p)$-coin when one can flip multiple times
a $p$-coin, but $p$ is unknown \citep{keane1994bernoulli}. More
specifically, letting $\mathcal{P}\subseteq[0,1]$ and $f:\mathcal{P}\rightarrow[0,1]$,
such a factory must output a flip of a $f(p)$-coin for any $p\in\mathcal{P}$
without knowledge of $p$ but is allowed to flip a $p$-coin an almost
surely finite number of times during its execution. One early example
of such a factory was presented in \citet{von1951various}, where
$f(p)=1/2$. More recent interest arises from the perfect simulation
of general regenerative processes \citep{asmussen1992stationarity}.
This inspired \citet{keane1994bernoulli}, whose major contribution
is a necessary and sufficient condition for the existence of a Bernoulli
factory for $f$: that $f$ is either constant, or continuous and
satisfies, for some $n\geq1$,
\[
\min\left\{ f(p),1-f(p)\right\} \geq\min\left\{ p,1-p\right\} ^{n},\quad\forall p\in\mathcal{P}.
\]
It can be verified from this result that Bernoulli factories exist
for the functions given in Problems~\ref{prob:bf1} and~\ref{prob:bf2}
when $p\geq\beta>\epsilon>0$ for known constants $\beta$ and $\epsilon$.

While the existence of Bernoulli factories is shown in \citet{keane1994bernoulli},
the proof is not constructive. Bernoulli factory algorithms have been
provided in a series of papers \citep{nacu2005fast,latuszynski2011simulating,thomas2011practical,flegal2012exact,huber2013nearly}.
In these, most attention is paid to the Bernoulli factory for $f$
satisfying 
\begin{equation}
cp\leq\gamma\implies f(p)=cp\label{eq:standard_bf_problem}
\end{equation}
for a given $c>0$ and $\gamma\in(0,1)$. The values that $f$ takes
when $cp>\gamma$ can be thought of as an ``extension'', particularly
when it is known that the factory will only be invoked for a $p$-coin
satisfying $cp\leq\gamma$. For example, in \citet{nacu2005fast}
and \citet{latuszynski2011simulating} the function $f(p)=\min\left\{ \gamma,cp\right\} $
is treated, whose extension is thus the constant function $\gamma$.
\citet{flegal2012exact} propose an alternative extension so that
$f$ is twice differentiable, which provides some performance guarantees
(\citealp[Proposition~10]{nacu2005fast}, see also \citealp{holtz2011new}).
In \citet{huber2013nearly}, an algorithm is presented for which the
extension is not explicitly shown, and which requires an expected
number of flips of a $p$-coin which is bounded above by $9.5c/(1-\gamma)$
whenever $cp<\gamma$ \citep[Theorem~1]{huber2013nearly}. In this
paper all our algorithms are implemented using the efficient procedure
in \citet{huber2013nearly}.

Problem~\ref{prob:bf2} can be solved by a standard solution for
$f$ satisfying (\ref{eq:standard_bf_problem}); in Appendix~\ref{sec:An-alternative-solution}
we develop an alternative solution to Problem~\ref{prob:bf2} via
a generic solution to the sign problem, which may be of independent
interest. We now show that Problem~\ref{prob:bf1} can be solved
using any solution to Problem~\ref{prob:bf2}. 
\begin{prop}
\label{prop:e_over_p_coin_standard}One can flip an $\epsilon/p$-coin
where $p>\epsilon>0$ by sampling $K\sim{\rm Geometric}(\epsilon)$
and flipping a $\left\{ (1-p)/(1-\epsilon)\right\} ^{K-1}$-coin.
\end{prop}
Algorithm~\ref{alg:eoverpcoinalg} is an implementation of the procedure
in Proposition~\ref{prop:e_over_p_coin_standard}, where a $\left\{ (1-p)/(1-\epsilon)\right\} ^{K-1}$-coin
flip is viewed as the product of $K-1$ independent $(1-p)/(1-\epsilon)$-coin
flips. The procedure can terminate as soon as any of these $(1-p)/(1-\epsilon)$-coin
flips is $0$. We have therefore formulated Algorithm~\ref{alg:eoverpcoinalg}
as a ``race'' between an $\epsilon$-coin and a $(p-\epsilon)/(1-\epsilon)$-coin.
One can indeed directly check that in this algorithm, the probability
that $Y_{n}=1$ conditional upon $Y_{n}+Z_{n}\geq1$ is $\epsilon/p$.
We note that an alternative solution to Problem~\ref{prob:bf1} is
to simulate $K\sim{\rm Geometric}(\epsilon)$ and then simulate a
$c_{K-1}\tilde{p}_{K-1}$-coin where $c_{k}=(1-\epsilon)^{-k}$ and
$\tilde{p}_{k}=(1-p)^{k}$, but we do not pursue this further.

\begin{algorithm}[H]
\protect\caption{Simulate an $\epsilon/p$-coin flip using a $(1-p)/(1-\epsilon)$-coin\label{alg:eoverpcoinalg}}

For $n=1,2,\ldots$:
\begin{enumerate}
\item Simulate an $\epsilon$-coin flip, $Y_{n}$. If $Y_{n}=1$, stop and
output $1$.
\item Simulate a $(p-\epsilon)/(1-\epsilon)$-coin flip, $Z_{n}$. If $Z_{n}=1$,
stop and output $0$.\end{enumerate}
\end{algorithm}

\subsection{Computational cost of algorithms~\ref{alg:Practical_imputation}
and~\ref{alg:Practical_multigamma}\label{sub:Computional-cost-of}}

A striking relationship between Algorithms~\ref{alg:Practical_imputation}
and~\ref{alg:Practical_multigamma} is that they have almost exactly
the same expected computational effort, when the $\epsilon/p$-coin
flips in Algorithm~\ref{alg:Practical_imputation} are simulated
using Algorithm~\ref{alg:eoverpcoinalg}. 
\begin{prop}
\label{lem:eoverp_cost}The expected number of $(1-p)/(1-\epsilon)$-coins
required to simulate an $\epsilon/p$-coin using Algorithm~\ref{alg:eoverpcoinalg}
is $(1-\epsilon)/p$.
\end{prop}

\begin{cor}
\label{prop:split_cost}The expected number of $(1-p)/(1-\epsilon)$-coins
required to simulate $X_{n}$ and $\rho_{n-1}^{(a,\epsilon)}$ in
Algorithm~\ref{alg:Practical_imputation} is $1-\epsilon$.
\end{cor}

\begin{prop}
\label{prop:equiv_1mp_over_1me_coins_and_Pi_samples}Using either
Algorithm~\ref{alg:Practical_imputation} or~\ref{alg:Practical_multigamma},
the expected number of $(1-p)/(1-\epsilon)$-coin flips required to
simulate a single tour of the split chain $\tilde{{\bf X}}_{a,\epsilon}$
is $\epsilon^{-1}-1$, and the average number of samples from $\Pi$
to additionally simulate the tour itself is $\epsilon^{-1}$.
\end{prop}
Solving Problem~\ref{prob:bf2} directly using the algorithm of \citet{huber2013nearly}
involves taking $c=(1-\epsilon)^{-1}$ and $\gamma=(1-\beta)/(1-\epsilon)$
in (\ref{eq:standard_bf_problem}). We suggest that one takes $\beta\leq0.5$
and $\epsilon=\beta/2$ because of the following result, which indicates
that the expected number of $p$-coin flips required is then upper
bounded by a small constant using the default algorithm settings in
\citet{huber2013nearly}.

\begin{prop}
\label{prop:constant_huber}Let $\beta\leq0.5$, $\epsilon=\beta/2$
and $p\geq\beta$. The expected number of $p$-coin flips to produce
a $(1-p)/(1-\epsilon)$-coin flip using the algorithm of \citet{huber2013nearly}
with its default settings is bounded above by $11$.
\end{prop}
As a result, when $\beta\leq0.5$ and $\epsilon=\beta/2$, one is
ensured that the expected cost of obtaining a perfect sample is bounded
above by $12/\epsilon$. Empirically we have observed that this bound
is fairly tight, but that the expected number of $p$-coin flips required
to produce a $(1-p)/(1-\epsilon)$-coin flip may even be inferior
to $7$.

\section{Introduction of an artificial singleton atom\label{sec:artificial_atom}}

\subsection{State space extension}

On general state spaces, the existence of accessible, singleton atoms
is not guaranteed. In most statistical applications, for example,
one has $\mathsf{X}=\mathbb{R}^{d}$ and $\pi$ admits a density w.r.t.\ the
Lebesgue measure, in which case no such atom for $\Pi$ can exist
since $\pi(\{a\})=0$ for any $a\in\mathsf{X}$. \citet{brockwell2005identification}
suggested in such cases to define a different transition kernel $\check{\Pi}$
on an extended state space for which such an \emph{artificial} singleton
atom exists. The atom is artificial in the sense that one defines
an extended state space $\check{\mathsf{X}}:=\mathsf{X}\cup\{a\}$
for some distinguished $a$, and seeks to design $\check{\Pi}$ such
that it is Harris recurrent and irreducible with unique invariant
probability measure $\check{\pi}:\mathcal{B}(\check{\mathsf{X}})\rightarrow[0,1]$
satisfying, for some $k\in(0,1)$,
\begin{equation}
\check{\pi}(A)=k\pi(A)+(1-k)\mathbb{I}(a\in A),\quad A\in\mathcal{B}(\check{\mathsf{X}}).\label{eq:extended_invariant}
\end{equation}
When (\ref{eq:extended_invariant}) holds, it follows that $\check{\pi}(A)=k\pi(A)$
for any $A\in\mathcal{B}(\mathsf{X})$. We denote by $\check{\mathbf{X}}:=(\check{X}_{n})_{n\geq1}$
the Markov chain with transition kernel $\check{\Pi}$. 

It is always possible to recover an irreducible, Harris recurrent
Markov chain with invariant probability measure $\pi$ from $\check{\mathbf{X}}$.
If we define $\check{\mathbf{Y}}:=(\check{X}_{n})_{n\in\mathcal{J}}$
with $\mathcal{J}:=\left\{ n:\check{X}_{n}\in\mathsf{X}\right\} $
then $\check{\mathbf{Y}}$ is the Markov chain $\check{\mathbf{X}}$
``watched'' in the set $\mathsf{X}$ and is irreducible and Harris
recurrent with invariant probability measure $\pi$ \citep[Theorem~1]{brockwell2005identification}.
One can also check that $\check{\mathbf{Y}}$ satisfies (\ref{eq:minorization_general})
with $s(x)=\check{\Pi}(x,\{a\})$ and $\nu({\rm d}y)=\check{\Pi}(a,{\rm d}y)\mathbb{I}(y\in\mathsf{X})/\check{\Pi}(a,\mathsf{X})$
\citep[Theorem~2]{brockwell2005identification}.

\subsection{Practical design of the distribution \textmd{\normalsize{}$\check{\pi}$
}and transition kernel $\check{\Pi}$\label{sub:Practical-design-of}}

We give here practical ways to define the distribution $\check{\pi}$
and kernel $\check{\Pi}$ given $\pi$ and $\Pi$. In many statistics
applications $\pi$ admits a density w.r.t.\ to a dominating measure
$\lambda$ on $\mathsf{X}$ and we can compute an unnormalized version
$\gamma(x)$ of this density. In this case, we can choose a $b>0$
and define an unnormalized version $\check{\gamma}(x)$ of the density
of $\check{\pi}$ w.r.t.\ the dominating measure $\lambda+\delta_{a}$
on $\check{\mathsf{X}}$ through 
\begin{equation}
\check{\gamma}(x):=\mathbb{I}(x\in\mathsf{X})\gamma(x)+\mathbb{I}(x=a)b.\label{eq:unnormalizedartificialtargetdensity}
\end{equation}
It follows from (\ref{eq:unnormalizedartificialtargetdensity}) that
$\check{\pi}({\rm d}x)=\check{\gamma}(x)\left\{ \lambda\left({\rm d}x\right)+\delta_{a}\left({\rm d}x\right)\right\} /\check{\gamma}(\mathsf{X})$
satisfies (\ref{eq:extended_invariant}) with $k=\{1+b/\gamma(\mathsf{X})\}^{-1}$.
In practice, it is desirable that $\check{\pi}(\{a\})$ be not too
close to either $0$ or $1$ so that visits to the artificial atom
are frequent and yet the Markov chain spends a substantial amount
of time in the set $\mathsf{X}$. This suggests that an estimate of
$\gamma(\mathsf{X})$ is necessary to be able to choose an appropriate
value of $b$.

Two strategies for designing $\check{\Pi}$ are proposed in \citet{brockwell2005identification}.
The first one is to let $\check{\Pi}$ be a Metropolis--Hastings transition
kernel with target $\check{\pi}$ and proposal distribution $\check{Q}_{x}$
where $x\in\check{\mathsf{X}}$. If $\Pi$ is itself a Metropolis--Hastings
kernel, it is possible to modify its proposal $Q_{x}$ where $x\in\mathsf{X}$
in order to design $\check{Q}_{x}$. For example, one could choose,
for some $w\in(0,1)$ and a probability measure $\mu:\mathcal{B}(\check{\mathsf{X}})\rightarrow[0,1]$,
\begin{equation}
\check{Q}_{x}\left({\rm d}y\right)=\mathbb{I}(x\in\mathsf{X})\left\{ wQ_{x}\left({\rm d}y\right)+(1-w)\delta_{a}\left({\rm d}y\right)\right\} +\mathbb{I}(x=a)\mu\left({\rm d}y\right).\label{eq:brockwell-MH}
\end{equation}

The second strategy is to define, for some $w\in(0,1)$ and transition
kernels $\Pi_{1}$ and $\Pi_{2}$, 
\[
\check{\Pi}(x,{\rm d}y):=w\Pi_{1}(x,{\rm d}y)+(1-w)\Pi_{2}(x,{\rm d}y),
\]
where $\Pi_{1}(x,{\rm d}y)=\mathbb{I}(x\in\mathsf{X})\Pi(x,{\rm d}y)+\mathbb{I}(x=a)\delta_{a}({\rm d}y)$
and $\Pi_{2}$ allows the chain to move between $\mathsf{X}$ and
$\{a\}$. One choice of $\Pi_{2}$, suggested by \citet{brockwell2005identification},
is a Metropolis--Hastings kernel with proposal
\[
Q_{x}\left({\rm d}y\right)=\mathbb{I}(x\in\mathsf{X})\delta_{a}\left({\rm d}y\right)+\mathbb{I}(x=a)\mu\left({\rm d}y\right).
\]

For both strategies, the choices of $w$ and $\mu$ can greatly affect
performance. \citet{brockwell2005identification} suggest that adaptation
of these parameters at regeneration times of $\check{\mathbf{Y}}$
can be beneficial, a procedure justified theoretically by \citet{gilks1998adaptive}.
The strategies outlined here are not exhaustive. In Section~\ref{sec:icsmc_atom}
a different approach is taken, but with essentially the same idea
of modifying $\Pi$ in an appropriate way to define $\check{\Pi}$.

\subsection{General applicability of the methodology\label{sub:General-applicability-of}}

The perfect simulation algorithms we have developed rely on the uniform
ergodicity of the Markov chain of interest. We show here that whenever
$\mathbf{X}$ is uniformly ergodic, then $\check{{\bf X}}$ is uniformly
ergodic when $\check{\Pi}$ is a modification of $\Pi$ in a specific
sense.
\begin{prop}
\label{prop:ue-inheritance}Assume that a generic Markov kernel $\check{\Pi}:\check{\mathsf{X}}\times\mathcal{B}(\check{\mathsf{X}})\rightarrow[0,1]$
satisfies $\check{\Pi}(a,\mathsf{X})>0$ and for some $w>0$,
\begin{equation}
\check{\Pi}(x,A)\geq w\Pi(x,A),\quad x\in\mathsf{X},A\in\mathcal{B}(\mathsf{X}).\label{eq:check_pi_modification_definition}
\end{equation}
Then $\mathbf{X}$ being uniformly ergodic implies that $\check{\mathbf{X}}$
is uniformly ergodic but the converse does not hold.
\end{prop}
The conditions of Proposition~\ref{prop:ue-inheritance} are met
for both strategies in Section~\ref{sub:Practical-design-of}. This
implies that one can modify the transition kernels of uniformly ergodic
Markov chains on a general state space to construct transition kernels
of uniformly ergodic Markov chains with an artificial singleton atom
straightforwardly. Inspection of the proof of the converse failing
to hold also suggests that in some cases one can even obtain perfect
samples from the stationary distribution of a non-uniformly ergodic
Markov chain by suitable definition of $\check{\Pi}$, but we do not
pursue this further here.

The existence of a $\beta>0$ such that $\underline{p}\geq\beta$
is guaranteed in general for uniformly ergodic Markov chains with
a proper, accessible atom, as long as one is willing to consider a
$k$-step transition kernel. Indeed, any such chain $\check{\mathbf{X}}$
with a proper, accessible atom $\{a\}$ and transition kernel $\check{\Pi}$
satisfies, for some $m\in\mathbb{N}$ and $d>0$,
\[
\inf_{x\in\check{\mathsf{X}}}\check{\Pi}^{m}(x,\{a\})\geq d,
\]
and $\lim_{k\rightarrow\infty}\inf_{x\in\check{\mathsf{X}}}\check{\Pi}^{k}(x,\{a\})=\check{\pi}(\{a\})$.
Therefore, one can increase $\underline{p}$ by choosing a large enough
$k$ and treating the Markov kernel $\check{\Pi}^{k}$ as the Markov
kernel of interest. This presents no additional difficulties, since
we only require the ability to simulate from $\check{\Pi}^{k}$ and
identify the atom $\{a\}$. When $\{a\}$ is an artificial atom, it
may also be necessary to construct $\check{\Pi}$ in such a way that
$\check{\pi}(\{a\})$ is not very close to $0$ or $1$. The existence
of $\beta>0$ does not of course imply that it is always known in
practice but, at least in principle, one can simulate from the stationary
distribution of any uniformly ergodic Markov chain $\mathbf{X}$ via
the introduction of an artificial atom and the consideration of a
$k$-step version of the modified transition kernel. In Section~\ref{sec:icsmc_atom},
a Markov kernel $\check{P}_{N}$ with an artificial atom is considered
for which one can always treat $\check{P}_{N}$ itself as the Markov
kernel of interest. Moreover, the subscript $N$ is a parameter of
this Markov kernel that controls its rate of convergence.

The uniform ergodicity requirement for application of Algorithms~\ref{alg:Practical_imputation}
or~\ref{alg:Practical_multigamma} might appear strong but a similar
condition is required for the general CFTP algorithm presented in
\citet{propp1996exact}, as established in \citep[Theorem~4.2]{foss1998perfect}.

\section{Perfect simulation according to a Feynman--Kac path measure\label{sec:icsmc_atom}}

\subsection{Feynman--Kac path measure\label{sub:Feynman--Kac-path-measure}}

In this section, we focus on a generic discrete-time Feynman--Kac
model with time horizon $n$. Let $(\mathsf{Z},\mathcal{B}(\mathsf{Z}))$
be a measurable space. Consider a probability measure $\mu:\mathcal{B}(\mathsf{Z})\rightarrow[0,1]$,
some Markov kernels $M_{p}:\mathsf{Z}\times\mathcal{B}(\mathsf{Z})\rightarrow[0,1]$
for $p\in\{2,\ldots,n\}$ and non-negative $\mathcal{B}(\mathsf{Z})$-measurable
functions $G_{p}:\mathsf{Z}\rightarrow\mathbb{R}_{+}$ for $p\in\{1,\ldots,n\}$.
We write $\mathcal{M}:=(M_{p})_{p\in\{2,\ldots,n\}}$ and $\mathcal{G}:=(G_{p})_{p\in\{1,\ldots,n\}}$.
We define for any $p\in\{1,\ldots,n\}$, the measure $\gamma_{p}$
by
\begin{equation}
\gamma_{p}(A):=\int_{A}\left[\prod_{q=1}^{p}G_{q}(z_{q})\right]\mu({\rm d}z_{1})\prod_{q=2}^{p}M_{q}(z_{q-1},{\rm d}z_{q}),\quad A\in\mathcal{B}(\mathsf{Z}^{p}),\label{eq:unnorm_gamma_p}
\end{equation}
and its associated probability measure $\pi_{p}:=\gamma_{p}(1)^{-1}\gamma_{p}$.
With $\mathsf{X}:=\mathsf{Z}^{n}$ the Feynman--Kac path measure of
interest is the probability measure $\pi:=\pi_{n}$ on $\mathcal{B}(\mathsf{X})$.

Feynman--Kac models naturally accommodate HMMs \citep[see, e.g.,][]{DelMoral2004}.
Indeed, consider a latent general state space Markov chain $(Z_{p})_{p\geq1}$
such that $Z_{1}\sim\mu$ and $\left.Z_{p}\right|Z_{p-1}\sim M_{p}\left(Z_{p-1},\cdot\right)$
and observations $(Y_{1},\ldots,Y_{n})=(y_{1},\ldots,y_{n})$ where
$(Y_{1},\ldots,Y_{n})$ are assumed conditionally independent given
$(Z_{1},\ldots,Z_{n})$ with each $Y_{p}$ depending on $(Z_{1},\ldots,Z_{n})$
only through $Z_{p}$ with $\left.Y_{p}\right|Z_{p}\sim g_{p}\left(Z_{p},\cdot\right)$.
If we let $G_{p}(z_{p})=g_{p}(z_{p},y_{p})$, then $\pi$ corresponds
to the distribution of the states $(Z_{1},\ldots,Z_{n})$ of the unobserved
Markov chain conditional upon the observations $(y_{1},\ldots,y_{n})$
and $\pi(f)$ is the associated conditional expectation of $f(Z_{1},\ldots,Z_{n})$.

More generally, Feynman--Kac models, through expressions such as (\ref{eq:unnorm_gamma_p}),
can be used to define arbitrary distributions $\pi$. The primary
benefit of expressing distributions of interest in this way is that
it allows the approximation of expectations, $\pi(f)$, using SMC
methods. For example, a sophisticated SMC methodology for sampling
from a complex distribution using a sequence of auxiliary ``bridging''
distributions and associated MCMC kernels is the SMC sampler methodology
of \citet{DelMoral2006}.

\subsection{A particle filter\label{sub:A-particle-filter}}

The general purpose SMC algorithm, often referred to as a particle
filter, for estimating $\pi(f)$ was proposed in \citet{stewart1992use},
\citet{Gordon1993} and \citet{Kitagawa1996} in the context of HMMs;
see \citet{Doucet2008} for a recent survey. The algorithm is described
in Algorithm~\ref{alg:SMC}, in which $\mathcal{C}(p_{1},\ldots,p_{N})$
generically denotes the categorical distribution over $\{1,\ldots,N\}$
with probabilities proportional to $(p_{1},\ldots,p_{N})$. We follow
the presentation of \citet{Andrieu2010}. For each $i\in\{1,\ldots,N\}$
and $p\in\{1,\ldots,n\}$, the variable $\zeta_{p}^{i}$ is the $i$th
\emph{particle} at time $p$; when $p>1$, $A_{p-1}^{i}$ is the index
of the \emph{ancestor} of this same particle, in the sense that in
Algorithm~\ref{alg:SMC}, 
\[
\zeta_{p}^{i}\sim M_{p}(\zeta_{p-1}^{A_{p-1}^{i}},\cdot).
\]
This particular instance of SMC uses multinomial resampling, so-called
because the ancestor indices $\{A_{p-1}^{i}:i\in\{1,\ldots,N\}\}$
are conditionally i.i.d.\ given $(\zeta_{p-1}^{1},\ldots,\zeta_{p-1}^{N})$
and distributed according $\mathcal{C}\left(G_{p-1}(\zeta_{p-1}^{1}),\ldots,G_{p-1}(\zeta_{p-1}^{N})\right)$,
implying that the ``offspring'' vector 
\[
\left(\sum_{i=1}^{N}\mathbb{I}(A_{p-1}^{i}=1),\ldots,\sum_{i=1}^{N}\mathbb{I}(A_{p-1}^{i}=N)\right)
\]
is multinomially distributed.

\begin{algorithm}[H]
\protect\caption{SMC: sampling $V\sim Q^{N}$\label{alg:SMC}}

\begin{enumerate}
\item Simulate $\zeta_{1}^{i}\sim\mu$ for $i\in\{1,\ldots,N\}$.
\item {[}If $n>1${]} For $p=2,\ldots,n$

\begin{enumerate}
\item Simulate $A_{p-1}^{i}\sim\mathcal{C}\left(G_{p-1}(\zeta_{p-1}^{1}),\ldots,G_{p-1}(\zeta_{p-1}^{N})\right)$
for $i\in\{1,\ldots,N\}$.
\item Simulate $\zeta_{p}^{i}\sim M_{p}(\zeta_{p-1}^{A_{p-1}^{i}},\cdot)$
for $i\in\{1,\ldots,N\}$.
\end{enumerate}
\item Set $V=(\zeta_{1}^{1},\ldots,\zeta_{n}^{N},A_{1}^{1},\ldots,A_{n-1}^{N})$.\end{enumerate}
\end{algorithm}

In order to define an estimate of the path measure $\pi$ as a function
of the random variables produced in Algorithm~\ref{alg:SMC}, it
is necessary to define the \emph{ancestral lineage} of a particle
index. Following \citet{Andrieu2010}, for any $k\in\{1,\ldots,N\}$,
we define $B^{k}$ to be the $\{1,\ldots,N\}^{n}$-valued random variable
satisfying $B_{n}^{k}:=k$ and $B_{p}^{k}:=A_{p}^{B_{p+1}^{k}}$.
The random variable
\begin{equation}
\zeta^{k}:=(\zeta_{1}^{B_{1}^{k}},\ldots,\zeta_{n}^{B_{n}^{k}})\label{eq:zeta_k}
\end{equation}
is then a path taking values in $\mathsf{X}$. With this notation,
a natural estimate of $\pi$ is 
\begin{equation}
\pi^{N}({\rm d}x):=\frac{\sum_{k=1}^{N}G_{n}(\zeta_{n}^{k})\delta_{\zeta^{k}}({\rm d}x)}{\sum_{j=1}^{N}G_{n}(\zeta_{n}^{j})},\label{eq:piN}
\end{equation}
and we estimate accordingly $\pi(f)$ using $\pi^{N}(f)$. Many asymptotic
properties of $\pi^{N}(f)$ have been established in the literature
\citep[see][for a comprehensive treatment]{DelMoral2004}. For example,
under very mild conditions $\pi^{N}(f)$ converges almost surely to
$\pi(f)$ and a $\sqrt{N}$-central limit theorem holds for $\pi^{N}(f)$
as $N\rightarrow\infty$ \citep{dmmiclo,Chopin2004}.

\subsection{Motivation for perfect simulation\label{sub:Motivation-for-perfect}}

The convergence results mentioned above establish that (\ref{eq:piN})
can be used to approximate the probability measure $\pi$. In fact,
it is clear that if $\pi^{N}(A)$ is ``close'' to $\pi(A)$ for
any $A\in\mathcal{B}(\mathsf{X})$, then the marginal distribution
of a single path $\zeta^{K}$ where $K\sim\mathcal{C}\left(G_{n}(\zeta_{n}^{1}),\ldots,G_{n}(\zeta_{n}^{N})\right)$
is ``close'' to $\pi$. We will see that under mild conditions,
this ``closeness'' can be controlled by choosing a large enough
value of $N$.

Letting $[N]:=\{1,\ldots,N\}$, we define $\mathsf{V}_{N}:=\mathsf{Z}^{Nn}\times[N]^{N(n-1)}$
and $Q^{N}:\mathcal{B}(\mathsf{V}_{N})\rightarrow[0,1]$ to be the
probability measure associated with the variables produced by running
the SMC scheme with $N$ particles described in Algorithm~\ref{alg:SMC}.
We then define a collection of conditional probability measures $\{Q_{v}^{N}:v\in\mathsf{V}_{N}\}$
such that $Q_{v}^{N}:\mathcal{B}([N]^{n}\times\mathsf{X})\rightarrow[0,1]$
is the probability measure associated with picking an ancestral lineage
and a path. We show how to sample from this measure in Algorithm~\ref{alg:Pick-a-particle}.
A joint probability measure $Q^{N}:\mathcal{B}([N]^{n}\times\mathsf{X}\times\mathsf{V}_{N})\rightarrow[0,1]$
is then defined as $Q^{N}({\bf k},{\rm d}x,{\rm d}v):=Q^{N}({\rm d}v)Q_{v}^{N}({\bf k},{\rm d}x)$
and the marginal distribution of $X$ when $(X,V)\sim Q^{N}$ is also
denoted by $Q^{N}:\mathcal{B}(\mathsf{X})\rightarrow[0,1]$ with $Q^{N}({\rm d}x):=\sum_{{\bf k}\in[N]^{n}}\int_{\mathsf{V}_{N}}Q^{N}({\bf k},{\rm d}x,{\rm d}v)$.

\begin{algorithm}[H]
\protect\caption{Pick a path: sampling $({\bf K},X)\sim Q_{v}^{N}$\label{alg:Pick-a-particle}}

Letting $v=(z_{1}^{1},\ldots,z_{n}^{N},a_{1}^{1},\ldots,a_{n-1}^{N})$,
\begin{enumerate}
\item Sample $K_{n}\sim\mathcal{C}\left(G_{n}(z_{n}^{1}),\ldots,G_{n}(z_{n}^{N})\right)$
and {[}if $n>1${]} for $p=n-1,n-2,\ldots,1$ set $K_{p}=a_{p}^{K_{p+1}}$.
\item Set ${\bf K}=(K_{1},\ldots,K_{n})$ and $X=\left(z_{1}^{K_{1}},\ldots,z_{n}^{K_{n}}\right)$.\end{enumerate}
\end{algorithm}

In order to show that $Q^{N}$ is indeed ``close'' to $\pi$, and
because the construction itself is a central component of the $\pi$-invariant
Markov kernel we shall be concerned with, we introduce a collection
of probability measures $\{\bar{Q}_{x}^{N}:x\in\mathsf{X}\}$. For
a given $x\in\mathsf{X}$, called the reference path, $\bar{Q}_{x}^{N}:\mathcal{B}([N]^{n}\times\mathsf{V}_{N})\rightarrow[0,1]$
corresponds to the distribution of the ancestral lineage of the reference
path and the $\mathsf{V}_{N}$-valued variables produced in a \emph{conditional}
SMC algorithm, a procedure introduced in \citet{Andrieu2010} and
described in Algorithm~\ref{alg:Conditional-SMC} that has both theoretical
and methodological applications. We will also denote by $\bar{Q}_{x}^{N}:\mathcal{B}(\mathsf{V}_{N})\rightarrow[0,1]$
the probability measure associated with the marginal distribution
of the $\mathsf{V}_{N}$-valued variables produced, i.e. $\bar{Q}_{x}^{N}({\rm d}v):=\sum_{{\bf k}\in[N]^{n}}\bar{Q}_{x}^{N}({\bf k},{\rm d}v)$.

\begin{algorithm}[H]
\protect\caption{Conditional SMC: sampling $({\bf K},V)\sim\bar{Q}_{x}^{N}$\label{alg:Conditional-SMC}}

\begin{enumerate}
\item Sample $K_{p}$ independently and uniformly on $\{1,\ldots,N\}$ for
each $p\in\{1,\ldots,n\}$.
\item Set $\zeta_{1}^{K_{1}}=x_{1}$ and simulate $\zeta_{1}^{i}\sim\mu$
for $i\in\{1,\ldots,N\}\setminus\{K_{1}\}$.
\item {[}If $n\geq1${]} For $p=2,\ldots,n$

\begin{enumerate}
\item Set $A_{p-1}^{K_{p}}=K_{p-1}$ and simulate $A_{p-1}^{i}\sim\mathcal{C}\left(G_{p-1}(\zeta_{p-1}^{1}),\ldots,G_{p-1}(\zeta_{p-1}^{N})\right)$
for $i\in\{1,\ldots,N\}\setminus\{K_{p}\}$.
\item Set $\zeta_{p}^{K_{p}}=x_{p}$ and simulate $\zeta_{p}^{i}\sim M_{p}(\zeta_{p-1}^{A_{p-1}^{i}},\cdot)$
for $i\in\{1,\ldots,N\}\setminus\{K_{p}\}$.
\end{enumerate}
\item Set ${\bf K}=(K_{1},\ldots,K_{n})$ and $V=(\zeta_{1}^{1},\ldots,\zeta_{n}^{N},A_{1}^{1},\ldots,A_{n-1}^{N})$.\end{enumerate}
\end{algorithm}

Central to some of the methodology in \citet{Andrieu2010} is the
behaviour of a random variable known as the normalizing constant estimate,
renormalized so that it has expectation $1$ when evaluated at $({\bf K},X,V)\sim Q^{N}$
\citep{Andrieu2009,Andrieu2012,Leea,andrieu2013uniform}. Letting
$v=(z_{1}^{1},\ldots,z_{n}^{N},a_{1}^{1},\ldots,a_{n-1}^{N})$, this
renormalized estimate is the Radon--Nikodym derivative 
\begin{equation}
\phi^{N}(v):=\frac{\pi({\rm d}x)\bar{Q}_{x}^{N}({\bf k},{\rm d}v)}{Q^{N}({\bf k},{\rm d}x,{\rm d}v)}=\frac{1}{\gamma_{n}(1)}\prod_{p=1}^{n}\frac{1}{N}\sum_{j=1}^{N}G_{p}(z_{p}^{j}),\label{eq:RN-derive-NC-estimate}
\end{equation}
which depends neither on the value of the picked particle $X$ nor
on its ancestral lineage ${\bf K}$. Letting $\mathsf{E}^{N}$ denote
expectation w.r.t.\ the law of $Q^{N}$ and $\bar{\mathsf{E}}_{x}^{N}$
denote expectation w.r.t.\ the law of $\bar{Q}_{x}^{N}$, the following
result provides a non-asymptotic upper bound on the Radon--Nikodym
derivative between $\pi$ and $Q^{N}$.
\begin{prop}
\label{prop:RNderiv_ub}Assume that for each $p\in\{1,\ldots,n\}$,
$G_{p}(z_{p})>0$ for all $z_{p}\in\mathsf{Z}$. Then 
\[
\frac{\pi({\rm d}x)}{Q^{N}({\rm d}x)}\leq\int_{\mathsf{V}_{N}}\phi^{N}(v)\bar{Q}_{x}^{N}({\rm d}v)=\bar{\mathsf{E}}_{x}^{N}\left[\phi^{N}(V)\right].
\]
\end{prop}
\begin{rem*}
A similar, but slightly more complicated expression holds without
the strict positivity of the potentials.
\end{rem*}
Proposition~\ref{prop:RNderiv_ub}, together with results in \citet[Section~4]{andrieu2013uniform}
implies a uniform upper bound on $\pi({\rm d}x)/Q^{N}({\rm d}x)$
when the potentials in $\mathcal{G}$ are additionally uniformly bounded
above. In practice, SMC methods work well for models which have ``forgetting
properties''. This notion can be quantified, and used to bound quantities
that arise in their analysis. In particular we define for each $p\in\{1,\ldots,n\}$
and $k\in\{0,\ldots n-p\}$ the function $F_{p,k}:\mathsf{Z}\rightarrow\mathbb{R}_{+}$
by
\[
F_{p,k}(z_{p}):=\frac{\gamma_{p-1}(1)}{\gamma_{p+k}(1)}\int_{\mathsf{Z}^{k}}\left[\prod_{q=p}^{p+k}G_{q}(z_{q})\right]\left[\prod_{q=p+1}^{p+k}M_{q}(z_{q-1},{\rm d}z_{q})\right],
\]
with the convention that $\gamma_{0}(1)=1$. The value $F_{p,k}(z_{p})$
can be viewed as an appropriately normalized expectation of the product
of the $k+1$ potentials evaluated along a path started at $Z_{p}=z_{p}$
and evolving according to the dynamics in $\mathcal{M}$ for $k$
steps. In the context of an HMM, e.g., $F_{p,k}(z_{p})$ is the likelihood
of the $k+1$ observations $(y_{p},\ldots,y_{p+k})$ conditional upon
$Z_{p}=z_{p}$ divided by the likelihood of the same observations
conditional upon $y_{1},\ldots,y_{p-1}$, and indeed one can see that
if $\sup_{z_{p}\in\mathsf{Z}}F_{p,k}(z_{p})$ eventually stabilizes
as $k$ increases then the HMM ``forgets'' in a particular sense
the distribution of $Z_{p}$ at time $p+k$. The level of forgetting
of the finite time Feynman--Kac model as a whole can be crudely, but
succinctly summarized in the single quantity
\[
F:=\max_{p\in\{1,\ldots,n\},k\in\{0,\ldots,n-p\}}\sup_{z_{p}\in\mathsf{Z}}F_{p,k}(z_{p}),
\]
which is necessarily finite whenever the potentials in $\mathcal{G}$
are upper bounded.
\begin{prop}
\label{prop:RN_pi_QN}Assume there exists $B<\infty$ such that for
each $p\in\{1,\ldots,n\}$, $0<G_{p}(z_{p})<B$ for all $z_{p}\in\mathsf{Z}$.
Then for any $N\geq2$, 
\begin{equation}
\sup_{x\in\mathsf{X}}\frac{\pi({\rm d}x)}{Q^{N}({\rm d}x)}\leq\bar{\mathsf{E}}_{x}^{N}\left[\phi^{N}(V)\right]\leq\left(1+\frac{2(F-1)}{N}\right)^{n}.\label{eq:RNderivate-bound-smc}
\end{equation}

\end{prop}
The second inequality in (\ref{eq:RNderivate-bound-smc}) is tight.
In interpreting (\ref{eq:RNderivate-bound-smc}), it is informative
to consider the effect of doubling the number of particles $N$. Letting
$C=2(F-1)$, we have $\left(1+\frac{C}{N}\right)^{n}/\left(1+\frac{C}{2N}\right)^{n}=\left(\frac{2N+2C}{2N+C}\right)^{n}=\left(1+\frac{C}{2N+C}\right)^{n}$
which is approximately $2^{n}$ if $C\gg2N$ and $1+\frac{nC}{2N}$
when $2N\gg C$. Hence the improvement is close to exponential for
small $N$ but asymptotically linear.

While in general $F$ will grow exponentially with $n$ when the potentials
have a uniform upper bound, under stronger assumptions it can be shown
that $F$ has an upper bound that is independent of $n$. In practice,
the performance of these algorithms in a large class of statistical
applications suggests that $\bar{\mathsf{E}}_{x}^{N}\left[\phi^{N}(V)\right]$
grows as a small power of $n$, and there is some recent theory pointing
towards this direction \citep[Section~4.3]{lindsten_pg}. One can
think of $F$ as providing a rough indication of the level of difficulty
of the Feynman-Kac model, in terms of the closeness of $Q^{N}$ to
$\pi$, and more generally as quantifying the suitability of SMC methodology
for a given application.

The form of (\ref{eq:RNderivate-bound-smc}) immediately implies that
rejection sampling to obtain perfect samples according to $\pi$ is
possible \emph{in principle}. However, since it is not possible to
evaluate $\pi({\rm d}x)/Q^{N}({\rm d}x)$ exactly such a procedure
has no existing practical implementation. A natural way to bypass
the calculation of $\pi({\rm d}x)/Q^{N}({\rm d}x)$, consists of using
rejection sampling to sample from the joint distribution defined by
$\pi({\rm d}x)\bar{Q}_{x}^{N}({\rm d}v)$, with proposal $Q^{N}({\rm d}x,{\rm d}v)$.
One can accept proposals with probability $M_{N}^{-1}\phi^{N}(V)$,
where $M_{N}\geq\sup_{v}\phi^{N}(v)$. From (\ref{eq:RN-derive-NC-estimate})
we can see that $M=\gamma_{n}(1)^{-1}\prod_{p=1}^{n}\sup_{x}G_{p}(x)$,
which is independent of $N$ and typically grows exponentially in
$n$, is usually the smallest value satisfying this condition. The
expected acceptance probability $\mathsf{E}^{N}\left[M^{-1}\phi^{N}(V)\right]=M^{-1}$
is unfortunately independent of $N$ and so this algorithm appears
unable to take advantage of (\ref{eq:RNderivate-bound-smc}).

The combination of the potential of Algorithm~\ref{alg:SMC} to produce
perfect samples and the inability to evaluate $\pi({\rm d}x)/Q^{N}({\rm d}x)$
invites modifications of Algorithm~\ref{alg:SMC}. In \citet{RubenthalerPS},
a perfect sampling method is proposed where the mechanism governing
particle offspring is fundamentally changed from selection with a
constant population size at each time to stochastic branching; their
method is able to produce perfect samples from some models (see also
Section~\ref{sub:Remarks}). We will instead use a $\pi$-invariant
Markov kernel proposed in \citet{Andrieu2010} whose main ingredient
is Algorithm~\ref{alg:Conditional-SMC} and is, in many ways, closer
in spirit to Algorithm~\ref{alg:SMC}.

\subsection{Iterated conditional SMC kernel\label{sub:Iterated-conditional-SMC}}

The iterated conditional SMC kernel \citep{Andrieu2010} is a Markov
kernel $P_{N}:\mathsf{X}\times\mathcal{B}(\mathsf{X})\rightarrow[0,1]$
defined by 
\[
P_{N}(x,{\rm d}y):=\int_{\mathsf{V}_{N}}\bar{Q}_{x}^{N}({\rm d}v)Q_{v}^{N}({\rm d}y).
\]
That is, we first run a conditional SMC algorithm (Algorithm~\ref{alg:Conditional-SMC})
to produce the auxiliary variables $V$ and then choose a path (Algorithm~\ref{alg:Pick-a-particle})
given the auxiliary variables $V$ just produced. This Markov kernel
leaves $\pi$ invariant, see, e.g., \citet[Lemma~4]{andrieu2013uniform}.

Our interest in using this particular Markov kernel to facilitate
regenerations and perfect simulation stems from both its good empirical
performance and its uniform ergodicity properties. In particular,
\citet{chopin:singh:2013} provided the first such result implying
that under some assumptions $P_{N}(x,{\rm d}y)\geq\epsilon_{N}\pi({\rm d}y)$
where $\epsilon_{N}\rightarrow1$ but without a rate of convergence.
\citet{andrieu2013uniform} and \citet{lindsten_pg} provide convergence
under the weaker assumption that the potentials in $\mathcal{G}$
are $\pi$-essentially upper bounded with quantitative rates of convergence
for $\epsilon_{N}\rightarrow1$. In addition, \citet{andrieu2013uniform}
showed that this essential boundedness condition is also necessary
for $P_{N}$ to define a uniformly ergodic Markov chain.

\citet[Corollary~15]{andrieu2013uniform} provides the following bound
on $\epsilon_{N}$ in terms of $F$
\begin{equation}
P_{N}(x,{\rm d}y)\geq\left(\frac{N-1}{N+2(F-1)}\right)^{n}\pi({\rm d}y)=\left(1-\frac{2F-1}{N+2(F-1)}\right)^{n}\pi({\rm d}y),\label{eq:icsmc_minorization}
\end{equation}
which implies that $\epsilon_{N}>0$ whenever $F$ is finite and $N\geq2$.
This bound can be compared with (\ref{eq:RNderivate-bound-smc}),
and behaves asymptotically in $N$ as $(1-\frac{(2F-1)n}{N})+\mathcal{O}(N^{-2})$
but for large $n$ and large $F$ the improvement by increasing $N$
is drastic. Finally, under fairly strong assumptions on the elements
of $\mathcal{\mu}$, $\mathcal{M}$ and $\mathcal{G}$, $F$ does
not grow with $n$ and it is possible to control $\epsilon_{N}$ as
$n\rightarrow\infty$ by scaling $N$ only linearly with $n$. In
particular, if $N=Bn$ then $\epsilon_{N}\geq\exp\left(-\frac{2(F-1)}{B}\right)$
\citep[Corollary~15]{andrieu2013uniform}. In order to guarantee $\epsilon_{N}\geq0.75$,
e.g., we can take $B=-2(F-1)/\log(0.75)$, which is slightly less
than $7(F-1)$ and so ensuring that $\epsilon_{N}$ is fairly large
requires a computational complexity in $\mathcal{O}(n^{2})$.

\subsection{Atomic extensions of the Feynman--Kac path measure and iterated conditional
SMC kernel\label{sub:Atomic-extensions-of}}

Let $\check{\mathsf{Z}}:=\mathsf{Z}\cup\alpha$, where $\alpha=\{a\}$
and $a$ is a distinguished point, and $\check{\mathsf{X}}:=\check{\mathsf{Z}}^{n}$,
$a_{n}:=\left(a,\ldots,a\right)$. We propose here a generic way to
define a new probability measure $\check{\pi}$ on $\check{\mathsf{X}}$
which satisfies for some $k\in(0,1)$,
\begin{equation}
\check{\pi}(A)=k\pi(A)+(1-k)\mathbb{I}(a_{n}\in A),\quad A\in\mathcal{B}(\check{\mathsf{X}}),\label{eq:extended_invariant-FeynmanKac}
\end{equation}
where $\pi$ is the Feynman--Kac path measure. We do not follow the
approach suggested in Section~\ref{sec:artificial_atom} to build
$\check{\pi}$ but directly modify the original Feynman--Kac model
to introduce the atom $\alpha_{n}:=\{a_{n}\}$. The benefit of this
approach is that it allows us to come up straightforwardly with a
$\check{\pi}$-invariant Markov kernel with well-understood properties.

The extended Feynman--Kac model is defined by the initial distribution
$\check{\mu}$, the Markov kernels $\check{\mathcal{M}}:=(\check{M}_{p})_{p\in\{2,\ldots,n\}}$
and potential functions $\check{\mathcal{G}}:=(\check{G}_{p})_{p\in\{1,\ldots,n\}}$
on $\check{\mathsf{Z}}$ which are given by
\begin{eqnarray*}
 &  & \check{\mu}(A):=(1-b)\mu(A\cap\mathsf{Z})+b\mathbb{I}\{a\in A\},\\
 &  & \check{M}_{p}(x,A):=M_{p}(x,A)\mathbb{I}\{x\in\mathsf{X}\}+\mathbb{I}\{x=a,A=\alpha\},\\
 &  & \check{G}_{p}(x):=G_{p}(x)\mathbb{I}\{x\in\mathsf{X}\}+\psi_{p}\mathbb{I}\{x=a\},
\end{eqnarray*}

for $A\in\mathcal{B}(\check{\mathsf{Z}})$ where $b\in(0,1)$ and
$\psi_{1},\ldots,\psi_{n}$ are user-defined positive constants. Similar
to the definition of $\gamma_{n}$ in Section~\ref{sub:Feynman--Kac-path-measure},
we define the measure $\check{\gamma}_{n}$ by
\[
\check{\gamma}_{n}(A):=\int_{A}\left[\prod_{p=1}^{n}\check{G}_{p}(x_{p})\right]\check{\mu}({\rm d}x_{1})\prod_{p=2}^{n}\check{M}_{p}(x_{p-1},{\rm d}x_{p}),\quad A\in\mathcal{B}(\check{\mathsf{X}}),
\]
and its associated probability measure $\check{\pi}:=\check{\gamma}_{n}(1)^{-1}\check{\gamma}_{n}$.
It follows easily from these definitions that (\ref{eq:extended_invariant-FeynmanKac})
holds with 
\begin{equation}
k=\frac{1-b}{1-b+b\gamma_{n}(1)^{-1}\prod_{p=1}^{n}\psi_{p}}.\label{eq:valuek_icsmc}
\end{equation}

The transition kernel we propose to use to sample from the extended
Feynman--Kac path measure $\check{\pi}$ is simply the i-cSMC kernel
targetting $\check{\pi}$ instead of $\pi$. We denote by $\check{P}_{N}$
this kernel, where $N$ is the number of particles used. The uniform
ergodicity of the Markov chain with transition kernel $\check{P}_{N}$
is straightforward, and indeed the Markov chain defined by $\check{P}_{N}$
is uniformly ergodic if and only if the Markov chain defined by $P_{N}$
is uniformly ergodic. The forgetting properties of the extended Feynman--Kac
model can be quantified via
\[
\check{F}_{p,k}(z_{p}):=\frac{\check{\gamma}_{p-1}(1)}{\check{\gamma}_{p+k}(1)}\int_{\check{\mathsf{Z}}^{k}}\left[\prod_{q=p}^{p+k}\check{G}_{q}(z_{q})\right]\left[\prod_{q=p+1}^{p+k}\check{M}_{q}(z_{q-1},{\rm d}z_{q})\right],
\]
and
\[
\check{F}:=\max_{p\in\{1,\ldots,n\},k\in\{0,\ldots,n-p\}}\sup_{z_{p}\in\mathsf{Z}}\check{F}_{p,k}(z_{p}).
\]

At this point, it is clear that the notion of forgetting encoded by
$F$ and $\check{F}$ is not the same as forgetting of the initial
distribution of a Feynman--Kac model, which refers to the time $n$
marginal of $\pi_{n}$ being asymptotically independent of $\mu$
as $n\rightarrow\infty$ under suitable assumptions. Indeed, one
can see that conditional upon $Z_{p}=a$, the distribution of $Z_{p+k}$
is a point mass at $a$, regardless of $k$, and yet we can still
have ``forgetting'' in this weaker sense.

\begin{prop}
\label{prop:icsmc-artificial-minorization}The Markov chain with transition
kernel $\check{P}_{N}$ is uniformly ergodic if and only if the potentials
in $\mathcal{G}$ are $\pi$-essentially bounded and $N\geq2$. Moreover,
for any $N\geq2$, 
\begin{equation}
\check{P}_{N}(x,A)\geq\check{\epsilon}_{N}\check{\pi}(A),\quad x\in\check{\mathsf{X}},A\in\mathcal{B}(\check{\mathsf{X}}),\label{eq:minorization_Pcheck}
\end{equation}
where $\check{\epsilon}_{N}>0$ and $\check{\epsilon}_{N}\rightarrow1$
as $N\rightarrow\infty$. One can take $\check{\epsilon}_{N}=\left(\frac{N-1}{N+2(\check{F}-1)}\right)^{n}$.
\end{prop}
In order to implement either Algorithm~\ref{alg:psim_onestep_generic_split}
or Algorithm~\ref{alg:psim_onestep_generic}, one needs a lower bound
on $\inf_{x\in\check{\mathsf{X}}}\check{P}_{N}(x,\alpha_{n})$. This
can be obtained via a lower bound on the quantities $\check{\pi}(\alpha_{n})$
and a value of $\check{\epsilon}_{N}$ satisfying (\ref{eq:minorization_Pcheck}).
In general, it may not be possible to obtain such bounds analytically
using current theory, although Proposition~\ref{prop:icsmc-artificial-minorization}
is reassuring in that one knows that $\check{\epsilon}_{N}$ can be
made arbitrarily close to $1$ by increasing $N$. We will see in
the next section that a specific choice of the constants $\psi_{1},\ldots,\psi_{n}$
ensures that $\check{F}$ is equal to $F$, and while this choice
is typically not possible one can select $\psi_{1},\ldots,\psi_{n}$
stochastically so that they are close to this specific choice.

\subsection{Parameter settings\label{sub:Parameter-tuning}}

As mentioned earlier in Section~\ref{sec:artificial_atom}, it is
desirable that $\check{\pi}\left(\alpha_{n}\right)$ is neither too
close to $0$ nor too close to $1$. We outline here a simple strategy
for selecting $b$ and $\psi_{1},\ldots,\psi_{n}$ to achieve this.
We note that if $\prod_{p=1}^{n}\psi_{p}=\gamma_{n}(1)$ then we obtain
$\check{\gamma}_{n}(1)=\gamma_{n}(1)$ and $k=1-b$ in (\ref{eq:valuek_icsmc}).
We therefore suggest choosing $b=\frac{1}{2}$ with the constants
$\psi_{1},\ldots,\psi_{n}$ chosen such that $\prod_{p=1}^{n}\psi_{p}\approx\gamma_{n}(1)$.
The behaviour of the i-cSMC Markov chain associated with the extended
model specified by $\check{\mu}$, $\check{\mathcal{M}}$ and $\check{\mathcal{G}}$
can be related to the behaviour of the i-cSMC Markov chain associated
with the model specified by $\mu$, $\mathcal{M}$ and $\mathcal{G}$.
In particular, we can relate the forgetting properties of both models
by relating the constants $F$ and $\check{F}$ associated with the
two models. 
\begin{prop}
\label{prop:Fcheck_F}Assume that for some $E\geq1$ 
\[
\max_{p\in\{1,\ldots.n\}}\max\left\{ \frac{\psi_{p}}{\gamma_{p}(1)/\gamma_{p-1}(1)},\frac{\gamma_{p}(1)/\gamma_{p-1}(1)}{\psi_{p}}\right\} \leq E.
\]
Then $\check{F}\leq FE^{n}$.
\end{prop}
Although the dependence of this bound on $n$ is exponential, for
models in which $F$ is not too large, small values of $E$ result
with high probability when each $\psi_{p}$ is an SMC estimate of
$\gamma_{p}(1)/\gamma_{p-1}(1)$. Practically, we run Algorithm~\ref{alg:SMC}
with a large number $N'$ of particles for the model specified by
$\mu$, $\mathcal{M}$ and $\mathcal{G}$ and approximate each $\gamma_{p}(1)/\gamma_{p-1}(1)$
by $\psi_{p}:=\frac{1}{N'}\sum_{i=1}^{N'}G_{p}(\zeta_{p}^{i})$, which
is a consistent estimate as $N'\rightarrow\infty$. Then one can run
Algorithm~\ref{alg:SMC} with $N'$ particles for the model specified
by $\check{\mu}$, $\check{\mathcal{M}}$ and $\check{\mathcal{G}}$
and estimate $\check{\pi}(\alpha_{n})$ by $\frac{1}{N'}\sum_{i=1}^{N'}\mathbb{I}(\zeta_{n}^{i}=a)$.
Since these are only estimates, it is prudent to conservatively use
a value smaller than this latter estimate, e.g., by using a concentration
inequality after multiple runs of Algorithm~\ref{alg:SMC} for the
model specified by $\check{\mu}$, $\check{\mathcal{M}}$ and $\check{\mathcal{G}}$.
Ideally, this estimate is close to $b$, as suggested by (\ref{eq:valuek_icsmc}).
In practice, when $\check{\pi}(\alpha_{n})$ is close to $b$, this
is also a good indication that each $\psi_{p}$ is very close to $\gamma_{p}(1)/\gamma_{p-1}(1)$
since
\[
\gamma_{n}(1)=\prod_{p=1}^{n}\frac{\gamma_{p}(1)}{\gamma_{p-1}(1)}\approx\prod_{p=1}^{n}\psi_{p}.
\]

Approximating or bounding $F$ is challenging, but explicit upper
bounds can be obtained under strong assumptions (see Appendix~\ref{sec:Expressions-for}).
Although the finiteness of $F$ is guaranteed when $n$ is finite
whenever the potentials in $\mathcal{G}$ are bounded, little attention
has been paid in the literature to obtaining explicit upper bounds
that are reasonably small. For example, \citet{whiteley2013stability}
provides only the existence of $F<\infty$ as $n\rightarrow\infty$
but no explicit bounds. One situation in which an upper bound on $F$
can be sufficiently small to be of practical use is discussed in Section~\ref{sec:Applications}.

In practice, we have found that this procedure for selecting $\psi_{1},\ldots,\psi_{n}$
and estimating bounds on $\check{\epsilon}_{N}$ and $\check{\pi}(\alpha_{n})$
is very effective, even though deterministic bounds on $\check{\epsilon}_{N}$
and $\check{\pi}(\alpha_{n})$ are not available. Since this procedure
does not provide a lower bound on $\inf_{x\in\check{\mathsf{X}}}\check{P}_{N}(x,\alpha_{n})$,
one can only hope for a bound that holds with high probability. In
Sections~\ref{sub:Diagnostics-and-estimation} and~\ref{sub:Sensitivity-to-the}
we discuss diagnostics and estimation of $\inf_{x\in\check{\mathsf{X}}}\check{P}_{N}(x,\alpha_{n})$,
and the sensitivity of the procedure to choosing a value of $\beta$
superior to $\inf_{x\in\check{\mathsf{X}}}\check{P}_{N}(x,\alpha_{n})$.

\section{Discussion\label{sec:discussion}}

\subsection{Diagnostics and estimation of $\beta$\label{sub:Diagnostics-and-estimation}}

In some situations, it may not be possible to obtain a lower bound
$\beta$ on $\underline{p}$. One can nevertheless attempt to estimate
such a bound and there are many strategies one may use to perform
this estimation. In general, one seeks to find the maximum of the
function $g:\mathsf{X}\rightarrow[0,1]$ defined by $x\mapsto1-p(x)$,
and one can obtain unbiased estimates of $g(x)$ for any $x\in\mathsf{X}$
by sampling from $\Pi(x,\cdot)$. This suggests that a stochastic
optimization procedure could be used to estimate $\underline{p}$.
An alternative approach is to simulate the Markov chain $\mathbf{X}$
for a long period of time and use the chain to estimate $\underline{p}$.
It is possible then to use the same realization of the Markov chain
to impute the regeneration indicators, i.e., run Algorithm~\ref{alg:Practical_imputation}
retrospectively.

It is also possible to monitor the validity of the assumption that
a chosen $\beta$ satisfies $\beta\leq\underline{p}$ using a simple
and fairly inexpensive diagnostic. At each state $x$ visited during
the course of Algorithm~\ref{alg:Practical_imputation} or~\ref{alg:Practical_multigamma},
one can simply simulate $p(x)$-coins until their average exceeds
$\beta$. The following result shows that this strategy will require
little computational effort if $\underline{p}>\beta$ but that a failure
to terminate after a reasonable number of $p(x)$-coins have been
simulated indicates that $\underline{p}$ may be inferior to $\beta$. 
\begin{prop}
\label{prop:diagnostic}Let $p\in[0,1]$ and $(B_{i})_{i\geq1}$ be
i.i.d. ${\rm Bernoulli}(p)$ random variables. Let $\tau$ be the
stopping time defined by $\tau:=\inf\{n\geq1:\frac{1}{n}\sum_{i=1}^{n}B_{i}>\beta\}$.
Then the following results hold:
\begin{enumerate}
\item If $p<\beta$, then $\mathbb{P}(\tau<\infty)<1$.
\item If $p=\beta$, then $\mathbb{P}(\tau<\infty)=1$ and $\mathbb{E}(\tau)=\infty$.
\item If $p>\beta$, then $\mathbb{P}(\tau<\infty)=1$ and $\mathbb{E}(\tau)\leq(1-\beta)/(p-\beta)$.
\end{enumerate}
\end{prop}
These properties immediately imply that repeated use of Algorithms~\ref{alg:Practical_imputation}
or~\ref{alg:Practical_multigamma} in tandem with running the diagnostic
will not cause problems if $\beta<\underline{p}$ but that if $\beta=\pi-{\rm ess}\inf_{x\in\mathsf{X}}p(x)$
each call has an infinite expected running time and if $\beta>\pi-{\rm ess}\inf_{x\in\mathsf{X}}p(x)$
then with probability $1$, one of the calls will eventually fail
to terminate.
\begin{cor}
\label{cor:diagnostic_algorithm_implication}If the diagnostic is
performed at every state visited by Algorithm~\ref{alg:Practical_imputation}
or~\ref{alg:Practical_multigamma}, then:
\begin{enumerate}
\item If $\pi-{\rm ess}\inf_{x\in\mathsf{X}}p(x)<\beta$, then the algorithm
will not terminate with positive probability.
\item If $\pi-{\rm ess}\inf_{x\in\mathsf{X}}p(x)=\beta$, then the algorithm
has infinite expected running time.
\item If $\pi-{\rm ess}\inf_{x\in\mathsf{X}}p(x)>\beta$, then the algorithm
has finite expected running time.
\end{enumerate}
\end{cor}
In the special case where $\beta=1/m$ for some $m\in\mathbb{N}$,
one can express analytically $\mathbb{P}(\tau<\infty)$. The following
result shows that when $p$ is even slightly different to $\beta$
then the probability of not stopping in finite time is quite large.
For example, if $\beta=0.2$ and $p=0.19$ then the probability of
not stopping is over $0.06$. We note that one can also determine
probabilities associated with hitting times of certain subsets of
$[0,1]$ in this specific setting, but do not pursue this further.
\begin{prop}
\label{prop:skip_free_prob_never_positive}Let $p\in[0,\beta)$ where
$\beta=m^{-1}$ for some $m\in\mathbb{N}$, $(B_{i})_{i\geq1}$ be
i.i.d. ${\rm Bernoulli}(p)$ random variables and $\tau:=\inf\{n\geq1:\frac{1}{n}\sum_{i=1}^{n}B_{i}>\beta\}$.
Then $\mathbb{P}(\tau<\infty)=p(m-1)/(1-p)$.
\end{prop}
Finally, in Section~\ref{sub:Computional-cost-of} it was seen that
the computational complexity of the algorithms in Section~\ref{sub:indirect_s_e}
is improved when $\beta$ is not too close to $0$. In general, therefore,
it may be the case that $\underline{p}$ is too small for the algorithms
to be practical, even if $\mathbf{X}$ is uniformly ergodic. Following
the remarks in Section~\ref{sub:General-applicability-of}, one general
strategy is to treat the Markov kernel $\Pi^{k}$ as the Markov kernel
of interest. This can also be seen as an alternative to increasing
$N$ in the iterated conditional SMC kernel of Section~\ref{sub:Iterated-conditional-SMC},
and may be appropriate when $N$ is already sufficiently large that
$\epsilon_{N}$ in (\ref{eq:minorization_Pcheck}) is reasonably large.

\subsection{Sensitivity to the choice of $\epsilon$\label{sub:Sensitivity-to-the}}

Consider the general case where (\ref{eq:minorization_general}) holds
for $\nu=\delta_{a}$ and $\underline{p}=\underbar{\ensuremath{\epsilon}}>0$
but one has attempted to run Algorithm~\ref{alg:psim_onestep_generic}
with $s=\epsilon>\underbar{\ensuremath{\epsilon}}$. We investigate
here how this affects the samples obtained. It is clear that $R_{a,\epsilon}$
is not a Markov kernel but one can instead view the algorithm as sampling
from
\[
\tilde{R}_{a,\epsilon}(x,{\rm d}y):=\frac{\Pi(x,{\rm d}y)-\min\{\epsilon,p(x)\}\delta_{a}({\rm d}y)}{1-\min\{\epsilon,p(x)\}},
\]
which results from flipping a $\min\left\{ 1,[1-p(x)]/(1-\epsilon)\right\} $-coin
when deciding which mixture component in (\ref{eq:mixture_residual_epsilon})
to sample from. The corresponding Markov transition kernel whose invariant
distribution we obtain perfect samples from is
\[
\tilde{\Pi}(x,{\rm d}y)=\epsilon\delta_{a}({\rm d}y)+(1-\epsilon)\tilde{R}_{a,\epsilon}(x,{\rm d}y),
\]
which defines a uniformly ergodic Markov chain of invariant distribution
denoted $\tilde{\pi}$. The following result indicates that closeness
of $\epsilon$ and $\underbar{\ensuremath{\epsilon}}$ implies closeness
of $\pi$ and $\tilde{\pi}$.
\begin{prop}
\label{prop:sensitivity}The invariant distributions $\pi$ and $\tilde{\pi}$
satisfy
\[
\left\Vert \tilde{\pi}-\pi\right\Vert _{{\rm TV}}\leq1-\underbar{\ensuremath{\epsilon}}/\epsilon,
\]
where for a signed measure $\mu$ on $\mathcal{B}(\mathsf{X})$, $\left\Vert \mu\right\Vert _{{\rm TV}}:=\sup_{A\in\mathcal{B}(\mathsf{X})}|\mu(A)|$.
\end{prop}
The approximate Markov kernel $\tilde{R}_{a,\epsilon}(x,\cdot)$ is
not exactly what is sampled from in Algorithm~\ref{alg:Practical_multigamma}
when $p(x)<\epsilon$, as the latter is dependent on the specific
implementation of the Bernoulli factory of a $(1-p)/(1-\epsilon)$-coin.
However, Proposition~\ref{prop:sensitivity} provides a reasonable
approximation of the discrepancy between $\pi$ and the distribution
of the samples obtained from Algorithm~\ref{alg:Practical_multigamma}
when $\underline{p}<\epsilon$.

\subsection{Unbiased estimation of $\pi(f)$ with a perfect sample from $\pi$\label{sub:Unbiased-estimation-of}}

With a perfect sample $X_{1}\sim\pi$, it is trivial that $f(X_{1})$
is an unbiased estimate of $\pi(f)$, where $f:\mathsf{X}\rightarrow\mathbb{R}$.
To obtain a consistent estimate of $\pi(f)$, one can simply run any
irreducible, $\pi$-invariant Markov kernel $P$ initialized at $X_{1}$
and compute the unbiased MCMC estimate $n^{-1}\sum_{i=1}^{n}f(X_{i})$.
It may be less obvious that one can also combine a perfect sample
with a number of imperfect samples to obtain an unbiased and consistent
estimate of $\pi(f)$. Consider the $\pi$-invariant, i-cSMC kernel
of Section~\ref{sub:Iterated-conditional-SMC}. This kernel can be
expressed as 
\[
P_{N}(x,{\rm d}y):=\int_{\mathsf{X}^{N}}\tilde{P}_{N}(x,{\rm d}y_{1:N})\frac{\sum_{k=1}^{N}G(y_{k})\delta_{y_{k}}({\rm d}y)}{\sum_{j=1}^{N}G(y_{j})},
\]
where $\tilde{P}_{N}(x,{\rm d}y_{1:N}):=\int_{\mathsf{V}_{N}}\bar{Q}_{x}^{N}({\rm d}v)\prod_{k=1}^{N}\delta_{\zeta^{k}}({\rm d}y_{k})$,
with $\zeta^{k}$ defined in (\ref{eq:zeta_k}), and for $z\in\mathsf{X}=\mathsf{Z}^{n}$,
$G(z)=G_{n}(z_{n})$. It is clear that $P_{N}$ satisfies $\int_{\mathsf{X}}\pi({\rm d}x)P_{N}(x,{\rm d}z)=\pi({\rm d}z)$,
that $\int_{\mathsf{X}^{2}}f(y)\pi({\rm d}x)P_{N}(x,{\rm d}y)=\pi(f)$
and, therefore, that
\[
\int_{\mathsf{X}^{N+1}}\pi({\rm d}x)\tilde{P}_{N}(x,{\rm d}y_{1:N})\frac{\sum_{k=1}^{N}G(y_{k})f(y_{k})}{\sum_{j=1}^{N}G(y_{j})}=\pi(f).
\]
Hence, sampling $(Y_{1},\ldots,Y_{N})\sim\tilde{P}_{N}(X_{1},\cdot)$
implies that the ``self-normalized'' estimate 
\begin{equation}
\frac{\sum_{k=1}^{N}G(Y_{k})f(Y_{k})}{\sum_{j=1}^{N}G(Y_{j})}\label{eq:self_norm_unbiased}
\end{equation}
is an unbiased estimate of $\pi(f)$. The estimate (\ref{eq:self_norm_unbiased})
can be viewed as a specific case of a type of estimate studied in
\citet[Section~3]{whiteley2014twisted} (which are in general, however,
not unbiased), and a $\sqrt{N}$-central limit theorem holds for (\ref{eq:self_norm_unbiased})
under fairly mild assumptions. Moreover, the asymptotic variance is
identical to that associated with $\pi^{N}(f)$, the standard particle
filter estimate in Section~\ref{sub:A-particle-filter} \citep[Remark~7]{whiteley2014twisted}.

\subsection{Non-uniformly ergodic $\Pi$ and unbiased estimation of $\pi(f)$\label{sub:Non-uniformly-ergodic-}}

In cases where $\underline{p}=0$ or $\Pi$ is not even uniformly
ergodic, the more generally applicable strategy of Algorithm~\ref{alg:psim_general_regen}
can still be used. Implementations of this algorithm have typically
sampled from $\eta_{n}^{(\nu,s)}$ in (\ref{eq:eta_n}) by rejection
using $\nu R_{\nu,s}^{n-1}$ as a proposal, which leads to an overall
expected infinite running time.

An alternative procedure, however, is apparent from viewing $\eta_{n}^{(\nu,s)}$
as a marginal distribution associated with a Feynman--Kac path measure.
That is, one can in principle sample from $\eta_{n}^{(\nu,s)}$ using
either of the algorithms in Section~\ref{sub:indirect_s_e} by constructing,
e.g., an i-cSMC kernel $P_{N}^{(n)}$ with an artificial atom. In
Section~\ref{sub:Iterated-conditional-SMC} it was seen that the
computational time required to obtain such a sample could be in $\mathcal{O}(n^{2})$,
at least under strong assumptions. More generally, we note that a
finite expected computational time of Algorithm~\ref{alg:psim_general_regen}
is guaranteed for a large class of Markov chains as long as one can
sample $N$ according to (\ref{eq:Ngeneraldist}) in finite expected
time and from $\eta_{n}^{(\nu,s)}$ in expected time polynomial in
$n$.
\begin{prop}
\label{prop:finite_expected_time_general_alg}Assume $\mathbb{P}_{\nu,s}(\tau_{\nu,s}\geq n)/\mathbb{E}_{\nu,s}(\tau_{\nu,s})=\mathcal{O}(n^{-b})$
for some $b>0$ and that one can sample according to (\ref{eq:Ngeneraldist})
in finite expected time. Then if the expected time to sample from
$\eta_{n}^{(\nu,s)}$ is in $\mathcal{O}(n^{d})$ for some $d>0$
then the expected time to complete Algorithm~\ref{alg:psim_general_regen}
is finite whenever $d<b-1$.
\end{prop}
While one can sample from $\eta_{n}^{(\nu,s)}$ in principle using
either of the algorithms in Section~\ref{sub:indirect_s_e} and an
appropriate sequence of i-cSMC kernels $(P_{N}^{(n)})_{n\geq1}$ with
corresponding constants $(\beta_{n})_{n\geq1}$, defining these sequences
in practice may constitute a research programme on their own. The
problem of being able to sample according to (\ref{eq:Ngeneraldist})
in finite expected time is a separate issue which has already been
addressed for a few Markov kernels in \citet{flegal2012exact}.

An alternative approach to Algorithm~\ref{alg:psim_general_regen}
when one can sample according to (\ref{eq:Ngeneraldist}) in finite
expected time, but is interested only in unbiased estimation of $\pi(f)$
for some $f:\mathsf{X}\rightarrow\mathbb{R}$ (as opposed to perfect
simulation) also follows from viewing $\eta_{n}^{(\nu,s)}$ as a marginal
distribution associated with a Feynman--Kac path measure. It is well
established that SMC algorithms such as Algorithm~\ref{alg:SMC}
produce unbiased estimates of unnormalized expectations $\gamma_{n}(f)$
\citep[Section~7.4.2]{DelMoral2004}, where $\gamma_{n}$ is defined
in Section~\ref{sub:Feynman--Kac-path-measure}. In this context,
this means that one can unbiasedly estimate $\gamma_{n}(f)=\mathbb{P}_{\nu,s}(\tau_{\nu,s}\geq n)\eta_{n}(f)$
using Algorithm~\ref{alg:SMC}. From the decomposition (\ref{eq:mixture_representation}),
we thus have the identity
\[
\pi(f)=\sum_{k\geq1}\frac{\mathbb{P}_{\nu,s}(\tau_{\nu,s}\geq k)}{\mathbb{E}_{\nu,s}(\tau_{\nu,s})}\eta_{k}(f)=\frac{1}{\mathbb{E}_{\nu,s}(\tau_{\nu,s})}\sum_{k\geq1}\gamma_{k}(f).
\]
Letting $g$ be a p.m.f.\ with support $\mathbb{N}$, Algorithm~\ref{alg:Unbiased-estimation-etan_smc}
provides an unbiased estimate of $\pi(f)$. In practice, one should
choose $g$ so that it has heavier tails than (\ref{eq:Ngeneraldist}).
The validity of the method is immediate upon viewing $\mathbb{I}(M=1)$
as a Bernoulli random variable with success probability
\[
\frac{\mathbb{P}_{\nu,s}(\tau_{\nu,s}\geq1)}{\mathbb{E}_{\nu,s}(\tau_{\nu,s})}=\frac{1}{\mathbb{E}_{\nu,s}(\tau_{\nu,s})},
\]
and $Z/g(K)$ as an importance sampling estimate of $\sum_{k\geq1}\gamma_{k}(f)$.
The relative variance of $\mathbb{I}(M=1)$ is therefore $\mathbb{E}_{\nu,s}(\tau_{\nu,s})\left[1-\mathbb{E}_{\nu,s}(\tau_{\nu,s})^{-1}\right]$
and that of $Z/g(K)$ will depend on both $g$ and the number of particles
$N$ used to estimate $\gamma_{K}(f)$ \citep[Section~9.4.1]{DelMoral2004}.

\begin{algorithm}[H]
\protect\caption{Unbiased estimation of $\pi(f)$ using SMC\label{alg:Unbiased-estimation-etan_smc}}

\begin{enumerate}
\item Simulate $M$ from the distribution with p.m.f. (\ref{eq:Ngeneraldist}).
If $M\neq1$, output $0$.
\item Sample $K\sim g$ and compute $Z$, an unbiased estimate of $\gamma_{K}(f)$
using SMC. Output $Z/g(K)$.\end{enumerate}
\end{algorithm}

\subsection{Parallel implementation of Markov chain Monte Carlo\label{sub:Parallel-implementation-of}}

Algorithms that can produce perfect samples in finite expected time
can immediately be implemented in parallel. However, simulation of
the split chain $\tilde{\mathbf{X}}_{\nu,s}$ is all that is required
for parallel implementation of MCMC, as noted by \citet{mykland1995regeneration}.
Indeed, if one can detect regeneration events, then one can simulate
multiple i.i.d.\  tours of the split chain in parallel and piece
them together afterwards. \citet{brockwell2005identification} introduced
artificial atoms for exactly this purpose, since detection of regeneration
is trivial for the split chain $\tilde{{\bf X}}_{a,p}$. 

Judicious choices of Markov kernel incorporating an artificial atom
are therefore highly pertinent in the current computing landscape,
where parallel algorithms offer significant computational advantages
\citep[see, e.g.,][]{Lee2010}. This type of regenerative parallelization
is complementary to parallel implementation of a sampling from a complex
Markov transition kernel, which typically require the use of a multiple
cores on a single computer, as it allows tours of such Markov chains
to be simulated on different computers. Notably, and in contrast to
the perfect simulation algorithms described in Section~\ref{sub:indirect_s_e},
there is no need for explicit knowledge of bounds on $p(x)$ as one
can simply use $s(x)=p(x)$ and simulate the split chain described
in Section~\ref{sub:MC_singleton_atoms}.

We now show that Markov chains possessing an identifiable atom can
be suitable for parallel implementation in the manner suggested without
stringent requirements on their rates of convergence. This is important
because in some cases, complex and intractable Markov kernels used
in statistical applications are not geometrically ergodic \citep[see, e.g.,][]{Andrieu2009,Andrieu2012,Leea}.
Since the lengths of each tour are i.i.d.\ random variables, we consider
here the length of the first tour $\tau_{a,p}$. From Kac's theorem,
we have $\mathbb{E}_{a,p}(\tau_{a,p})=\pi(\{a\})^{-1}<\infty$, and
for performance this suggests that having $\pi(\{a\})$ not too close
to $0$ or $1$ will ensure that at least the expected length of a
tour is small. In a parallel setting, however, it is important that
the expected length of the longest of $n$ independent tours be reasonably
short, and this can be guaranteed in some sense when $\mathbb{E}_{a,p}\left(\tau_{a,p}^{2}\right)<\infty$.

It follows indeed from \citet[Section~4.2]{david1970order} that when
$\mathbb{E}_{a,p}\left(\tau_{a,p}^{2}\right)<\infty$, 
\[
\mathbb{E}_{a,p}\left[\max\left\{ \tau_{1}^{(a,p)},\ldots,\tau_{n}^{(a,p)}\right\} \right]\leq\mathbb{E}_{a,p}\left(\tau_{a,p}\right)+(n-1)\left[\frac{\mathbb{E}_{a,p}\left(\tau_{a,p}^{2}\right)-\mathbb{E}_{a,p}\left(\tau_{a,p}\right)^{2}}{2n-1}\right]^{1/2},
\]
and so one can be assured that the expected length of the longest
tour grows at most as $\mathcal{O}(n^{1/2})$. In cases where the
tours are geometrically distributed with mean $\epsilon^{-1}$, one
can have much slower growth rates since then \citep[see, e.g, ][]{eisenberg2008expectation}
\[
\frac{1}{\lambda}H_{n}\leq\mathbb{E}_{a,p}\left[\max\left\{ \tau_{1}^{(a,p)},\ldots,\tau_{n}^{(a,p)}\right\} \right]\leq1+\frac{1}{\lambda}H_{n},
\]
where $\lambda=-\log(1-\epsilon)$ and $H_{n}=\sum_{k=1}^{n}k^{-1}$
is the $n$th harmonic number. Since $H_{n}$ grows as $\mathcal{O}(\log n)$,
one can expect very few long tours in these situations.

The following proposition, which is essentially a corollary of \citet[Theorem~1.1]{bednorz2008regeneration},
implies that if $\mathbf{X}$ is a reasonable Markov chain in a specific
sense then $\mathbb{E}_{\nu,p}\left(\tau_{\nu,p}^{2}\right)<\infty$.
\begin{prop}
\label{prop:reasonablechain-boundedvariancereturn}Let $\mathbf{X}$
be a time-homogeneous, Harris recurrent and aperiodic Markov chain
with unique stationary distribution $\pi$. If $\mathbf{X}$ possesses
a proper atom $\alpha$, then a $\sqrt{n}$-CLT holds for all bounded
functions $g:\mathsf{X}\rightarrow\mathbb{R}$ if and only if $\mathbb{E}_{\alpha,p}\left(\tau_{\alpha,p}^{2}\right)<\infty$.\end{prop}
\begin{rem*}
That $\mathbb{E}_{\alpha,p}\left(\tau_{\alpha,p}^{2}\right)<\infty$
implies a $\sqrt{n}$-CLT holds for all bounded functions $g:\mathsf{X}\rightarrow\mathbb{R}$
is provided by \citet[Theorem~4.3]{jarner2002polynomial}.
\end{rem*}

In particular, Proposition~\ref{prop:reasonablechain-boundedvariancereturn}
implies that if a $\sqrt{n}$-CLT holds for all bounded functions,
then the variance of the return time to $\alpha$ is finite. In a
statistical application, where $\pi$ is a posterior distribution,
this suggests that if the chain $\mathbf{X}$ is suitable for estimating
any posterior probability then $\mathbb{E}_{\alpha,p}\left(\tau_{\alpha,p}^{2}\right)$
is finite. As an example, if $\mathbf{X}$ is polynomially ergodic
and $\Pi$ satisfies jointly, for some function $V:\mathsf{X}\rightarrow[1,\infty)$,
constants $c,b\in(0,\infty)$ and $\gamma\in[1/2,1)$,
\[
\int_{\mathsf{X}}V(y)\Pi(x,{\rm d}y)\leq V(x)-cV(x)^{\gamma}+b\mathbf{1}_{\alpha}(x),\quad x\in\mathsf{X},
\]
then $\mathbb{E}_{\alpha,p}\left(\tau_{\alpha,p}^{2}\right)<\infty$
\citep[Section~4]{jarner2002polynomial}.

\subsection{Remarks\label{sub:Remarks}}

The main limitation of the proposed methodology is the requirement
that $\beta$ be chosen to be less than or equal to $\underline{p}$.
Nevertheless, the ability in principle to obtain perfect samples from
the path distribution associated with a Feynman--Kac model is encouraging,
as no previous algorithm was able to do so with guaranteed expected
computational time polynomial in the time horizon of the model. We
anticipate future progress in at least two areas: deterministic lower
bounds on $\underline{p}$ under reasonably weak assumptions and quantitative
bounds on the probability that $\beta>\underline{p}$ under suitable
assumptions on how $\Pi$ is defined, how $\beta$ is chosen and the
results of any diagnostics. As an example, a detailed analysis of
the suggested i-cSMC kernel construction of Section~\ref{sec:icsmc_atom}
may provide some assurances in practical applications.

An alternative perfect simulation algorithm has been proposed in \citet{RubenthalerPS}
to sample from Feynman--Kac path measures. Their algorithm relies
on a dominated CFTP procedure \citep{kendall2004geometric}  and
variants of Algorithm~\ref{alg:SMC} and Algorithm~\ref{alg:Conditional-SMC}
which use a branching mechanism instead of a multinomial resampling
mechanism to limit particle interactions. This algorithm has been
successfully applied to the simulation of self-avoiding random walks,
but computational guarantees have not yet been established. The approaches
developed here and in \citet{RubenthalerPS} are complementary.

The introduction of an artificial atom may even be beneficial in cases
where an atom does exist for $\Pi$. Consider the case where $\pi(\{a\})\leq\zeta$
for any $a\in\mathsf{X}$ in which case it follows that $\inf_{x\in\mathsf{X}}\Pi(x,\{a\})\leq\zeta$
for any $a\in\mathsf{X}$. If $\zeta$ is small then the perfect simulation
procedures we have provided are computationally expensive. Furthermore,
even if an atom is used only for parallel implementation of MCMC as
in Section~\ref{sub:Parallel-implementation-of}, the tours can be
prohibitively long when $\zeta$ is small since the expected return
time to an accessible atom $a$ is $\pi(\{a\})^{-1}$ by Kac's Theorem.

\section{Applications\label{sec:Applications}}

For simplicity, all of our applications are associated with discrete-time,
time-homogeneous HMMs with $\mathsf{Z}=\mathbb{R}$. That is, each
model is determined by the initial distribution of the latent states
$\mu$ and for each $p\in\{2,\ldots,n\}$, $M_{p}=M$. The observations
are encoded in the potential functions as $G_{p}(z_{p}):=g(z_{p},y_{p})$,
the likelihood of observation $Y_{p}=y_{p}$ conditional on $Z_{p}=z_{p}$.

In the first two applications, we make use of Corollary~\ref{cor:bound_on_A}
in Appendix~\ref{sec:Expressions-for}. In particular, we will show
that for some $A<\infty$, 
\begin{equation}
\sup_{z_{p-1}\in S_{p-1},z_{p}\in S_{p}}\frac{G_{p}(z_{p})}{\int_{\mathsf{Z}}G_{p}(z'_{p})M_{p}(z_{p-1},{\rm d}z'_{p})}\cdot\sup_{z_{p},z'_{p}\in S_{p},z_{p+1}\in S_{p+1}}\frac{M_{p+1}(z_{p},{\rm d}z_{p+1})}{M_{p+1}(z'_{p},{\rm d}z_{p+1})}\leq A,\label{eq:condition_A}
\end{equation}
where $S_{p}:=\{z\in\mathsf{Z}:G_{p}(z)>0\}$ for each $p\in\{1,\ldots,n\}$,
with $S_{0}:=\mathsf{Z}$. This allows us to take $F\leq A$ in (\ref{eq:icsmc_minorization}).
As suggested in Section~\ref{sub:Iterated-conditional-SMC}, we then
take $N=7(A-1)$ to guarantee that $\epsilon_{N}\geq0.75$.

Whenever (\ref{eq:condition_A}), or the more general condition (\ref{eq:condition_icsmcA})
in Appendix~\ref{sec:Expressions-for}, can be established in such
a way that $A$ is independent of $n$ then it follows that one only
needs a number of particles $N$ linear in $n$ in order to control
$\epsilon_{N}$. As the $N$ particles need to be propagated for $n$
time steps in Algorithm~\ref{alg:Conditional-SMC} the computational
cost per simulation from the i-cSMC Markov kernel is $\mathcal{O}(n^{2})$.
Since $\epsilon_{N}$ has been controlled, the expected number of
such simulations required to obtain a perfect sample is bounded by
a constant depending on $\beta$, and so the overall computational
complexity is $\mathcal{O}(n^{2})$.

\subsection{Particle in an absorbing medium\label{sub:Particle-in-an}}

We let $\mu$ be the uniform distribution on $[0,1]$ and $M(z,\cdot)$
be the Markov kernel associated with a normal distribution with mean
$z$ and variance $\sigma^{2}$, and whose density w.r.t.\ the Lebesgue
measure is denoted by $M(z,z')$. For some $S\in\mathcal{B}(\mathsf{X})$,
we let $G_{p}(z_{p})=\mathbf{1}_{S}(z_{p})$ for each $p\in\{1,\ldots,n\}$.
Under the dynamics associated with $\mu$ and $M$, $\pi$ is the
distribution of the path of a particle that has not left the set $S$.
This is a simple example of a ``hard obstacle'' model, which has
been studied in \citet{del2004particle}. The applicability of Algorithm~\ref{alg:SMC}
to obtain accurate estimates of $\pi(f)$ for this model has been
discussed in \citet{del2014particle}, where it is also remarked that
rejection sampling using the Markov chain defined by $\mu$ and $M$
alone has an expected computational cost that is $\gamma_{n}(1)^{-1}$,
which is exponentially increasing in $n$. Algorithm~\ref{alg:SMC},
however, does not produce samples exactly distributed according to
$\pi$.

As an example, we take $S=[0,1]$ and $\sigma^{2}=0.25$, from which
we can see that (\ref{eq:condition_A}) is satisfied with 
\[
A=\frac{1}{M(0,[0,1])}\frac{M(1,1)}{M(0,1)}<15.5.
\]

Since $\gamma_{n}(1)\leq M(0.5,S)^{n}=(0.683)^{n}$ it is clear that
for $n=100$ the expected number of samples for a rejection sampler
exceeds $10^{16}$. In contrast, $F\leq A$ suggests one can use $7(A-1)n\approx10^{4}$
particles to ensure $\epsilon_{N}>0.75$. As $n$ grows, the difference
becomes more and more significant; for example when $n=1000$, the
expected number of paths sampled in the rejection sampler is around
$10^{165}$ whereas $7(A-1)n\approx10^{5}$.

As suggested in Section~\ref{sub:Parameter-tuning}, Algorithm~\ref{alg:SMC}
was run to estimate the parameters $(\psi_{p})_{p\in\{1,\ldots,n\}}$
and $b$ was set to $0.5$. A number of simulations indicated that
$\check{P}_{N}((0.5,\ldots,0.5),\alpha_{n})$ was very likely to
be greater than $0.47$, and we chose $\beta=0.2$, and $\epsilon=0.1$
to be conservative.

We ran both algorithms to obtain perfect samples from $\pi$ with
$n=100$. The number of simulations from $\check{P}_{N}$ was in perfect
agreement with Propositions~\ref{prop:equiv_1mp_over_1me_coins_and_Pi_samples}
and~\ref{prop:constant_huber}. In our simulations this average cost
of flipping a $(1-p)/(1-\epsilon)$-coin in terms of simulations from
$\check{P}_{N}$ was around $5.5$ on average. This indicates an expected
number of simulations from $\check{P}_{N}$ to obtain a perfect sample
from $\check{\pi}$ of around $65$ and to obtain a perfect sample
from $\pi$ of about $130$, but it is worth noting that this figure
depends only on $\beta$ and $\epsilon$ and not $n$. Of course,
a less conservative estimate of $\beta$ would bring this number lower.
The diagnostic procedure of Section~\ref{sub:Diagnostics-and-estimation}
was run at each state visited by the Markov chain without issue.

Empirically, it seems that a smaller value $N=10^{4}$ was sufficient
to ensure $\inf_{x\in\check{\mathsf{X}}}\check{P}_{N}(x,\alpha_{n})$
greater than $0.4$ and it is true more generally that bounds such
as (\ref{eq:condition_A}) are naturally conservative. We note that
if one was to consider a larger set $S$ or a smaller variance $\sigma^{2}$
for the Markov kernels in $\mathcal{M}$ one could obtain smaller
values of $C$ by using Lemma~\ref{lemma:bound_on_A_alternative}
with a value of $m$ greater than $1$.

\subsection{Interval-censored sensor data\label{sub:Interval-censored-sensor-data}}

We consider a situation in which there are evenly spaced sensors which
track an object moving in $\mathbb{R}$. Each sensor $j\in\mathbb{Z}$
is associated with the interval of $[j,j+1)$ of $\mathbb{R}$. The
observations are thus integer valued and $Y_{p}=y_{p}$ indicates
that $Z_{p}\in[y_{p},y_{p+1})$. We let $\mu$ be a standard normal
distribution and $M(z,\cdot)$, as in Section~\ref{sub:Particle-in-an},
be the Markov kernel associated with a normal distribution with mean
$z$ and variance $\sigma^{2}$. In this case, we obtain a bound on
$A$ and therefore $C$ that is dependent on the data observed.

It is clear that $S_{p}:=[Y_{p},Y_{p}+1)$ for each $p\in\{1,\ldots,n\}$.

Assume $Y_{p}=i$ and $Y_{p+1}=j$. Then
\[
\sup_{z_{p},z'_{p}\in S_{p},z_{p+1}\in S_{p+1}}\frac{M(z_{p},{\rm d}z_{p+1})}{M(z'_{p},{\rm d}z_{p+1})}\leq\begin{cases}
\frac{M(i+1,j+1)}{M(i,j+1)} & i\leq j,\\
\frac{M(i,j)}{M(i+1,j)} & i>j.
\end{cases}
\]
Similarly, we have
\[
\inf_{z_{p-1}\in S_{p-1}}M_{p}(G_{p})(z_{p-1})=\begin{cases}
M(i,S_{p}) & i\leq j,\\
M(i+1,S_{p}) & i>j.
\end{cases}
\]
and $\bar{G}_{p}=\bar{G}=1$. Then, given observations $(Y_{1},\ldots,Y_{n})=(y_{1},\ldots,y_{n})$,
we bound $A$ by taking the two consecutive observations $y_{k}$
and $y_{k+1}$ that are furthest apart. We consider the case where
$y_{k}=i$ and $y_{k+1}=j$ and $i<j$, the cases with $i\geq j$
being analogous. Then we set
\[
A=\frac{1}{M(i,[j,j+1))}\frac{M(i+1,j)}{M(i,j+1)}.
\]

As above, Algorithm~\ref{alg:SMC} was run to estimate the parameters
$(\psi_{p})_{p\in\{1,\ldots,n\}}$ and $b$ was set to $0.5$. A number
of simulations indicated that $\inf_{x\in\check{\mathsf{X}}}\check{P}_{N}(x,\alpha_{n})$
was very likely to be greater than $0.5$, and we chose $\beta=0.2$,
and $\epsilon=0.1$ to be conservative. We ran both algorithms to
obtain perfect samples from $\check{\pi}$ with $n=100$ on a simulated
data set where $\max_{p\in\{1,\ldots,n-1\}}|y_{p}-y_{p+1}|=3$. This
led to $A=38$ with $\sigma^{2}=5$ and therefore $N=7(A-1)n=25900$.

The number of simulations from $\check{P}_{N}$ was again in perfect
agreement with Propositions~\ref{prop:equiv_1mp_over_1me_coins_and_Pi_samples}
and~\ref{prop:constant_huber}, with the average number of flips
of a $p$-coin to obtain a flip of a $(1-p)/(1-\epsilon)$-coin being
just over $5$. This indicates an expected number of simulations from
$\check{P}_{N}$ to obtain a perfect sample from $\pi$ of a little
over $120$, and again this figure depends only on $\beta$ and $\epsilon$
and not $n$. Of course, a less conservative estimate of $\beta$
would bring this number lower. The diagnostic procedure of Section~\ref{sub:Diagnostics-and-estimation}
was run at each state visited by the Markov chain without issue.

Our methodology is not only suited to indicator observations, but
we have used them here for simplicity. In addition, this is a setting
where $\mathsf{Z}$ is not compact, but each $\pi_{p}$ is compactly
supported with a support that is data dependent. This type of compact
support assumption is very much tied to obtaining bounds on $F$ via
(\ref{eq:condition_A}).

\subsection{Linear Gaussian state space model\label{sub:Linear-Gaussian-state}}

In practice, (\ref{eq:condition_A}) can fail to be satisfied for
any finite $A$, even though $F$ is finite whenever the potentials
are bounded. One such case is the linear Gaussian model, where $\mu$
is Gaussian, $M(z,\cdot)$ is Gaussian with a mean depending linearly
on $z$ and a variance independent of $z$ and $g(z,\cdot)$ is Gaussian
with a mean depending linearly on $z$ and a variance independent
of $z$. In such cases, $\pi$ is tractable and perfect samples may
be obtained from a Kalman smoother.

We simulated a time series of length $n=100$ in the case where $M(z,z')=\mathcal{N}(z';0.9z,1)$
and $g(z,y)=\mathcal{N}(y;z,1)$. After running Algorithm~\ref{alg:SMC}
with $10000$ particles to estimate the parameters $(\psi_{p})_{p\in\{1,\ldots,n\}}$
and setting $b$ to $0.5$, it was verified empirically that $\inf_{x\in\check{\mathsf{X}}}\check{P}_{N}(x,\alpha_{n})$
was likely to be greater than $0.5$ with $N=4096$. In addition,
comparison of each $\psi_{p}$ with the quantity $\gamma_{p}(1)/\gamma_{p-1}(1)$
computed using the Kalman filter showed relative discrepancies of
less than $0.02$ and $\gamma_{n}(1)^{-1}\prod_{p=1}^{n}\psi_{p}=1.01$
suggested the parameter settings were successful.

Running both perfect simulation algorithms with $\beta=0.2$ and $\epsilon=0.1$,
along with the diagnostic scheme in Section~\ref{sub:Diagnostics-and-estimation},
suggested that the samples obtained were perfect. Comparison of the
samples with $\pi$, obtained from the Kalman smoother, also did not
suggest any bias. As before, the computational expense was in perfect
agreement with Propositions~\ref{prop:equiv_1mp_over_1me_coins_and_Pi_samples}
and~\ref{prop:constant_huber} and on average less than $6$ $p$-coin
flips were needed to obtain a flip of a $(1-p)/(1-\epsilon)$-coin.

\subsection{Parallel particle marginal Metropolis--Hastings}

Our final application involves not perfect simulation, but parallel
implementation of the particle marginal Metropolis--Hastings (PMMH)
Markov chain introduced in \citet{Andrieu2010}. This Markov chain
is well-suited to estimating static parameters of an HMM, and we use
the same model and data as in Section~\ref{sub:Linear-Gaussian-state}.

Denoting by $\Theta=\mathbb{R}^{4}$ the parameter space with $\theta=(\theta_{1},\theta_{2},\theta_{3},\theta_{4})$
governing $M$ and $g$ by $M(z,z')=\mathcal{N}(z';\theta_{1}z,\theta_{2})$
and $g(z,y)=\mathcal{N}(y;\theta_{3}z,\theta_{4})$. The distribution
$\pi$ of interest is the posterior distribution of $\theta$ conditional
upon the observed data, and can be written as
\[
\pi(\theta):=\frac{\gamma_{\theta,n}(1)\varpi(\theta)}{\int\gamma_{\theta',n}(1)\varpi(\theta'){\rm d}\theta'},
\]
where $\varpi$ is a prior density for $\theta$, and $\gamma_{\theta,n}$
is the measure defined by (\ref{eq:unnorm_gamma_p}) for a given value
of the parameter $\theta$. More specifically, instead of a single
$\mu$, $\mathcal{M}$ and $\mathcal{G}$ there are now collections
$\{\mu_{\theta}:\theta\in\Theta\}$, $\{\mathcal{M}_{\theta}:\theta\in\Theta\}$,
and $\{\mathcal{G}_{\theta}:\theta\in\Theta\}$, and for each $\theta$
the triplet $(\mu_{\theta},\mathcal{M}_{\theta},\mathcal{G}_{\theta})$
specifies a Feynman--Kac model. We denote by $G_{\theta,p}$ the time
$p$ potential specified by $\mathcal{G}_{\theta}$. In this context,
since the Feynman--Kac models specify different versions of an HMM,
the quantity $\gamma_{\theta,n}(1)$ is equal to the marginal likelihood
of the observations conditional upon the parameter $\theta$.

While $\gamma_{\theta,n}(1)$ cannot be computed in practice, the
PMMH Markov chain is an example of a pseudo-marginal Markov chain
\citep{Beaumont2003,Andrieu2009} that in this context can be viewed
as evolving on $\Theta\times\mathbb{R}_{+}$ with transition kernel
\[
\Pi(\theta,w;{\rm d}\theta',{\rm d}w'):=\min\left\{ 1,\frac{\varpi(\theta')q(\theta',\theta)w'}{\varpi(\theta)q(\theta,\theta')w}\right\} U_{\theta'}^{N}({\rm d}w')q(\theta,{\rm d}\theta')+r(\theta,w)\delta_{(\theta,w)}({\rm d}\theta',{\rm d}w'),
\]
where
\[
r(\theta,w):=1-\int_{\Theta\times\mathbb{R}_{+}}\min\left\{ 1,\frac{\varpi(\theta')q(\theta',\theta)w'}{\varpi(\theta)q(\theta,\theta')w}\right\} q(\theta,{\rm d}\theta')U_{\theta'}^{N}({\rm d}w'),
\]
and $U_{\theta}^{N}$ is the distribution of the random variable 
\begin{equation}
\gamma_{\theta,n}^{N}(1):=\prod_{p=1}^{n}\frac{1}{N}\sum_{j=1}^{N}G_{\theta,p}(\zeta_{p}^{j}),\label{eq:nc_estimate_rv}
\end{equation}
when $(\zeta_{1}^{1},\ldots,\zeta_{n}^{N})\sim Q_{\theta}^{N}$, i.e.,
produced by running Algorithm~\ref{alg:SMC} with the Feynman--Kac
model specific to $\theta$. The marginal distribution of the $\theta$-coordinate
of the invariant distribution associated with this Markov kernel is
$\pi$. We consider the case where $q(\theta,\cdot)$ is a multivariate
normal distribution centred at $\theta$ with variance-covariance
matrix $0.1^{2}Id$, $\varpi$ is a multivariate normal distribution
centred at $0$ with variance-covariance matrix $0.2^{2}Id$, and
$N=2048$ to define $U_{\theta}^{N}$.

We modify $\Pi$ to produce $\check{\Pi}$ according to the first
strategy in Section~\ref{sub:Practical-design-of}. We first identify
a central parameter $\theta^{*}$, here taken to be $(0.9,1,1,1)$
for simplicity, and define the artificial atom to be $\alpha:=\{a\}$,
$a:=(a_{\theta},a_{w})$, where $a_{\theta}$ is a distinguished point
and $a_{w}$ is an estimate of $\gamma_{\theta^{*},n}(1)$ obtained
by running Algorithm~\ref{alg:SMC} and computing (\ref{eq:nc_estimate_rv})
with $N$ large. The prior density for $\theta$ is modified to $\check{\varpi}(\theta)=0.5\mathbb{I}(\theta\in\Theta)\varpi(\theta)+0.5\mathbb{I}(\theta=a_{\theta})$.
The ``re-entry'' proposal $\mu$ in (\ref{eq:brockwell-MH}) is
taken to be $\mu({\rm d}\theta,{\rm d}w)=q(\theta^{*},{\rm d}\theta)U_{\theta}^{N}({\rm d}w)$,
and the mixture weights in (\ref{eq:brockwell-MH}) are taken to be
$0.5$. Simulation of the split chain with $\nu=\delta_{a}$ and $s=p$
associated with the transition kernel $\check{\Pi}$ is then straightforward.

We ran the modified Markov chain $\check{\mathbf{X}}$ for $100000$
iterations, obtaining $71713$ tours. The proportion of time spent
at the artificial atom was therefore $0.72$ and the average length
of a tour was $1.4$. The empirical variance of the tour lengths was
$7.7$ and the length of the longest tour was $150$. There were only
$5$ tours whose length exceeded $100$. It is possible that the value
of $a_{w}$ was slightly too large, as almost three quarters of the
samples of $\check{\mathbf{X}}$ are discarded when $\check{\mathbf{Y}}$
is constructed. However, this factor of $4$ loss is offset by the
potential for a practically arbitrary gain in efficiency by distributing
the simulation of multiple tours across a number of different computers.

\section*{Acknowledgements}

The authors would like to thank the Isaac Newton Institute for Mathematical
Sciences, Cambridge, for support and hospitality during the programme
\emph{Advanced Monte Carlo Methods for Complex Inference Problems},
where work on this paper was undertaken. Arnaud Doucet's research
is supported by an Engineering and Physical Sciences Research Council
Established Career Research Fellowship.

\appendix

\section{Proofs}

\subsection{Proofs for Section~\ref{sec:Reg-PS-perfectatom}}
\begin{proof}[Proof of Proposition~\ref{prop:e_over_p_coin_standard}]
Since $p>\epsilon>0$, we have the Maclaurin series expansion $1/p=\sum_{k=0}^{\infty}(1-p)^{k}$
and so
\[
\frac{\epsilon}{p}=\epsilon\sum_{k=1}^{\infty}(1-p)^{k-1}=\sum_{k=1}^{\infty}\epsilon(1-\epsilon)^{k-1}\left(\frac{1-p}{1-\epsilon}\right)^{k-1},
\]
which is an expectation w.r.t.\ a ${\rm Geometric}(\epsilon)$ random
variable, and the result follows.
\end{proof}

\begin{proof}[Proof of Proposition~\ref{lem:eoverp_cost}]
We define the stopping time $\tau:=\inf\{n\geq1:Y_{n}+Z_{n}\geq1\}$
and it is clear that 
\begin{eqnarray*}
\mathbb{P}(Y_{\tau}=1) & = & \mathbb{P}(Y_{n}=1\mid Y_{n}+Z_{n}\geq1)\\
 & = & \frac{\epsilon}{\epsilon+(p-\epsilon)/(1-\epsilon)-\epsilon(p-\epsilon)/(1-\epsilon)}=\frac{\epsilon}{p}.
\end{eqnarray*}
We have $\mathbb{E}(\tau)=\left[1-(1-\epsilon)(1-p)/(1-\epsilon)\right]^{-1}=p^{-1}$
since $\tau$ is the minimum of two geometric random variables and
hence also geometrically distributed. Now when $Y_{\tau}=1$ it is
not necessary to flip $Z_{\tau}$, and so the expected number of $(p-\epsilon)/(1-\epsilon)$-coin
flips is
\[
\frac{\epsilon}{p}\left(\frac{1}{p}-1\right)+\left(1-\frac{\epsilon}{p}\right)\frac{1}{p}=\frac{1-\epsilon}{p}.
\]

\end{proof}

\begin{proof}[Proof of Corollary~\ref{prop:split_cost}]
We note that simulating $X_{n}$ and $\rho_{n-1}^{(a,\epsilon)}$
requires drawing $X_{n}\sim\Pi(X_{n-1},\cdot)$ and then, if $X_{n}=a$,
flipping an $\epsilon/p(X_{n-1})$-coin. From $p(X_{n-1})=\Pi(X_{n-1},\{a\})$
and Proposition~\ref{lem:eoverp_cost} we have that the expected
number of $(1-p(X_{n-1}))/(1-\epsilon)$-coin flips required to simulate
$(\rho_{n-1}^{(a,\epsilon)},X_{n})$ is $p(X_{n-1})(1-\epsilon)/p(X_{n-1})=1-\epsilon$.
\end{proof}

\begin{proof}[Proof of Proposition~\ref{prop:equiv_1mp_over_1me_coins_and_Pi_samples}]
For Algorithm~\ref{alg:Practical_imputation}, the number of samples
from $\Pi$ is $\epsilon^{-1}$ since we simulate $X_{2},\ldots,X_{\tau_{a,\epsilon}+1}$
each using $\Pi$ and the expected length of a tour is $\epsilon^{-1}$.
The expected number of $(1-p)/(1-\epsilon)$-coin flips required is
$(1-\epsilon)/\epsilon$ from Proposition~\ref{lem:eoverp_cost},
the fact that we simulate $\rho_{1}^{(a,\epsilon)},\ldots,\rho_{\tau_{a,\epsilon}}^{(a,\epsilon)}$
and the expected length of each tour being $\epsilon^{-1}$.

For Algorithm~\ref{alg:Practical_multigamma}, the number of $(1-p)/(1-\epsilon)$-coin
flips required is $\epsilon^{-1}-1$ since $N\sim{\rm Geometric}(\epsilon)$.
For each of the samples $X_{2},\ldots,X_{N}$ from $R_{a,\epsilon}(x,\cdot)$
the probability of selecting the mixture component $R_{a,p}$ in (\ref{eq:mixture_residual_epsilon})
is $(1-p(x))/(1-\epsilon)$ and the expected number of samples from
$\Pi(x,\cdot)$ to obtain a sample from $R_{a,p}(x,\cdot)$ by rejection
is $1/(1-p(x))$. Hence the expected number of samples from $\Pi$
for each of the $N-1$ steps is $\epsilon^{-1}$.
\end{proof}

The following Lemma is used to prove Proposition~\ref{prop:constant_huber}.
\begin{lem}
\label{lem:huber_bounded}Let $B\geq0.5$, $b\in[B,1]$, $p\in[0,b]$
and $C=2/(1+b)$, with $\varepsilon=1-Cb$. Then the expected number
of $p$-coin flips required to flip a $Cp$-coin using the algorithm
of \citet{huber2013nearly} with parameters $(C,\varepsilon)$ is
bounded above by $11$.\end{lem}
\begin{proof}
From inspection of the default algorithm settings, one can check that
since $B\geq0.217$ and $b\geq B$, one indeed has $\varepsilon=\min\{0.644,1-Cb\}=1-Cb$.
From \citet[Theorem~3.6]{huber2013nearly}, we have that the number
of $p$-coin flips $T$ satisfies
\begin{eqnarray*}
\mathbb{E}[T] & \leq & \frac{k(C-1)+C}{1-(Cp)^{k}}-\frac{C-1}{1-Cp}+\frac{\frac{r}{1-r}[\gamma k(\frac{C}{1-\varepsilon}-1)+(1-\gamma)^{2}\frac{C}{1-\varepsilon}]}{1-(Cp)^{k}}\\
 & \leq & \frac{1}{1-(Cp)^{k}}\left\{ k(C-1)+C+\frac{r}{1-r}\left[\gamma k(\frac{C}{1-\varepsilon}-1)+(1-\gamma)^{2}\frac{C}{1-\varepsilon}\right]\right\} 
\end{eqnarray*}
where $k:=2.3/(\gamma\varepsilon)$ and $r:=\exp(-2.3)/(1-\gamma)^{2}$
and $\gamma=1/2$. We first bound
\[
(Cp)^{k}\leq(Cb)^{k}=\left(\frac{2b}{1+b}\right)^{k}\nearrow\exp\left(-4.6\right)
\]
as $b\rightarrow1$ so $\frac{1}{1-(Cp)^{k}}\leq\frac{1}{1-\exp\left(-4.6\right)}<1.0102$.
Noting that $\varepsilon=C-1$ we have $k(C-1)+C=4.6+2/(1+b)$. Finally,
since $C/(1-\varepsilon)=b^{-1}$,
\[
\gamma k\left(\frac{C}{1-\varepsilon}-1\right)=\frac{2.3}{C-1}\left(\frac{1-b}{b}\right)=2.3\frac{1+b}{b},
\]
and $(1-\gamma)^{2}C/(1-\varepsilon)=1/(4b)$. All together, we obtain
\[
\mathbb{E}[T]\leq1.0102\left\{ 4.6+\frac{2}{1+b}+\frac{r}{1-r}\left[2.3\frac{1+b}{b}+\frac{1}{4b}\right]\right\} ,
\]
which is monotonically decreasing in $b$. Since $r/(1-r)\leq0.6696$
we have 
\[
\mathbb{E}[T]\leq1.0102\left\{ 4.6+\frac{2}{1+B}+0.6696\left[2.3\frac{1+B}{B}+\frac{1}{4B}\right]\right\} ,
\]
and for $B=0.5$ we obtain $\mathbb{E}[T]\leq11$.\end{proof}
\begin{rem*}
A bound can be obtained for any $B\geq0.217$ using the final displayed
equation. One could potentially derive better choices of $\gamma$
and $k$ based on the developments in the proof.\end{rem*}
\begin{proof}[Proof of Proposition~\ref{prop:constant_huber}]
This follows from Lemma~\ref{lem:huber_bounded} by noting that
\[
C=\frac{1}{1-\epsilon}=\frac{1}{1-\beta/2}=\frac{2}{2-\beta}=\frac{2}{1+b},
\]
with $b=1-\beta\geq0.5$, and $q=1-p\leq1-\beta=b$.
\end{proof}

\subsection{Proof for Section~\ref{sec:artificial_atom}}
\begin{proof}[Proof of Proposition~\ref{prop:ue-inheritance}]
Since $\mathbf{X}$ is uniformly ergodic, $\Pi$ satisfies for some
$m\in\mathbb{N}$, $\epsilon>0$ and probability measure $\nu$, 
\[
\Pi^{m}(x,A)\geq\epsilon\nu(A),\quad A\in\mathcal{B}(\mathsf{X}).
\]
We consider $\check{\Pi}^{m+1}$ and an arbitrary $A\in\mathcal{B}(\mathsf{X})$.
It follows that if $x\in\mathsf{X}$ then from (\ref{eq:check_pi_modification_definition}),
\[
\check{\Pi}^{m+1}(x,A)\geq\int_{\mathsf{X}}\check{\Pi}(x,{\rm d}y)\check{\Pi}^{m}(y,A)\geq w^{m+1}\epsilon\nu(A),
\]
since $\check{\Pi}(x,\mathsf{X})\geq w\Pi(x,\mathsf{X})=w$, and 
\[
\check{\Pi}^{m+1}(a,A)\geq\int_{\mathsf{X}}\check{\Pi}(a,{\rm d}y)\check{\Pi}^{m}(y,A)\geq\check{\Pi}(a,\mathsf{X})w^{m}\epsilon\nu(A).
\]
Hence $\inf_{x\in\check{\mathsf{X}}}\check{\Pi}^{m+1}(x,A)\geq\min\left\{ w,\check{\Pi}(a,\mathsf{X})\right\} w^{m}\epsilon\nu(A)$
and, as $\check{\Pi}(a,\mathsf{X})>0$ by assumption, it follows that
$\check{\mathbf{X}}$ is uniformly ergodic.

To show that the converse does not hold, consider the case where $\check{\Pi}=w\Pi_{1}+(1-w)\Pi_{2}$
is exactly as in the mixture strategy discussed in Section~\ref{sub:Practical-design-of}.
Then if $C:=\inf_{x\in\mathsf{X}}\mu({\rm d}x)/\pi({\rm d}x)>0$,
we have for $x\in\mathsf{X}$
\[
\check{\Pi}(x,\{a\})\geq(1-w)\min\left\{ 1,\check{\pi}(\{a\})\frac{\mu({\rm d}x)}{\check{\pi}({\rm d}x)}\right\} \geq(1-w)C\check{\pi}(\{a\})
\]
and $\check{\Pi}(a,\{a\})\geq w$. Hence, $\check{\Pi}(x,\{a\})\geq\min\left\{ w,(1-w)C\check{\pi}(\{a\})\right\} >0$
and so $\check{\mathbf{X}}$ is uniformly ergodic, irrespective of
the properties of $\Pi$. 
\end{proof}

\subsection{Proofs for Section~\ref{sec:icsmc_atom}}
\begin{proof}[Proof of Proposition~\ref{prop:RNderiv_ub}]
Since the potentials are strictly positive, the measures $Q^{N}({\rm d}x,{\rm d}v)$
and $\pi({\rm d}x)\bar{Q}_{x}({\rm d}v)$ are equivalent. Therefore,
\begin{eqnarray*}
Q^{N}({\rm d}x) & = & \sum_{{\bf k}\in[N]^{n}}\int_{\mathsf{V}_{N}}Q^{N}({\bf k},{\rm d}x,{\rm d}v)\\
 & = & \sum_{{\bf k}\in[N]^{n}}\int_{\mathsf{V}_{N}}\frac{Q^{N}({\bf k},{\rm d}x,{\rm d}v)}{\pi({\rm d}x)\bar{Q}_{x}^{N}({\bf k},{\rm d}v)}\pi({\rm d}x)\bar{Q}_{x}^{N}({\bf k},{\rm d}v)\\
 & = & \pi({\rm d}x)\int_{\mathsf{V}_{N}}\frac{1}{\phi^{N}(v)}\bar{Q}_{x}^{N}({\rm d}v)\\
 & \geq & \pi({\rm d}x)\left\{ \int_{\mathsf{V}_{N}}\phi^{N}(v)\bar{Q}_{x}^{N}({\rm d}v)\right\} ^{-1},
\end{eqnarray*}
by Jensen's inequality, and we conclude.
\end{proof}

\begin{proof}[Proof of Proposition~\ref{prop:RN_pi_QN}]
We provide here a simple proof leveraging a result of \citep{andrieu2013uniform}.
Under the assumptions given, condition (A1) in \citet{andrieu2013uniform}
always holds when $n$ is finite and their constant $\alpha$ is equal
to $F$ (see also Appendix~\ref{sec:Expressions-for}). Hence we
can apply \citet[Proposition~14]{andrieu2013uniform} to obtain that
for any $N\geq2$
\[
\sup_{x,y\in\mathsf{X}}\bar{\mathsf{E}}_{x,y}^{N}\left[\phi^{N}(V)\right]\leq\left(1+\frac{2(F-1)}{N}\right)^{n},
\]
where $\bar{\mathsf{E}}_{x,y}^{N}$ denotes expectation w.r.t.\ the
law of a doubly conditional SMC algorithm with two fixed paths $x,y\in\mathsf{X}$,
a construction defined in \citet{andrieu2013uniform}. Since
\[
\sup_{x\in\mathsf{X}}\bar{\mathsf{E}}_{x}^{N}\left[\phi^{N}(V)\right]\leq\sup_{x,y\in\mathsf{X}}\bar{\mathsf{E}}_{x,y}^{N}\left[\phi^{N}(V)\right],
\]
we conclude.
\end{proof}

\begin{proof}[Proof of Proposition~\ref{prop:icsmc-artificial-minorization}]
Since $\psi_{1},\ldots,\psi_{n}$ are constants, the potential functions
$\left\{ \check{G}_{p}:p\in\{1,\ldots,n\}\right\} $ are $\check{\pi}$-essentially
bounded above if and only if the potential functions $\{G_{p}:p\in\{1,\ldots,n\}\}$
are $\pi$-essentially bounded above. The result then follows from
\citet[Theorem~1]{andrieu2013uniform}.
\end{proof}

\begin{proof}[Proof of Proposition~\ref{prop:Fcheck_F}]
This follows immediately from Lemma~\ref{lem:Acheck} in Appendix~\ref{sec:Expressions-for}
\end{proof}

\subsection{Proofs for Section~\ref{sec:discussion}}
\begin{proof}[Proof of Proposition~\ref{prop:diagnostic}]
We define $(\xi_{i})_{i\geq1}$ by $\xi_{i}:=B_{i}-\beta$ and $S_{n}:=\sum_{i=1}^{n}\xi_{i}$
with $S_{0}:=0$. Then $\tau=\inf\{n\geq1:S_{n}>0\}$. For the first
two parts, we can think of $(S_{n})_{n\geq0}$ as a Markov chain on
$\mathbb{R}$ (\citealp{lamperti1960criteria}, see also \citealt[Chapter~9]{meyn2009markov})
with bounded increments. If $p<\beta$, $(S_{n})_{n\geq0}$ is transient
with negative drift and so its probability of return to $\mathbb{R}_{+}$
is strictly inferior to $1$. If $p=\beta$, $(S_{n})_{n\geq0}$ is
null recurrent and so its probability of return to $\mathbb{R}_{+}$
is $1$, and its expected return time to $\mathbb{R}_{+}$ is infinite.
Finally, when $p>\beta$, we can apply Wald's equation $\mathbb{E}(S_{\tau})=\mathbb{E}(\tau)\mathbb{E}(\xi_{1})$,
and since $S_{\tau}\leq1-\beta$ by construction and $\mathbb{E}(\xi_{1})=p-\beta$,
we have $\mathbb{E}(\tau)\leq(1-\beta)/(p-\beta)$.
\end{proof}

\begin{proof}[Proof of Corollary~\ref{cor:diagnostic_algorithm_implication}]
For the first part, let $A:=\{x\in\mathsf{X}:p(x)<\beta\}$ and by
assumption $\pi(A)>0$. Now let $n=\min\{k:\nu R_{a,\epsilon}^{k-1}(A)>0\}$,
which from the definitions of $\nu$ and $R_{a,\epsilon}$ must be
finite. Moreover, for any $k\in\{1,\ldots,n\}$, $\nu R_{a,\epsilon}^{k-1}$
is the distribution of $X_{k}$ using either algorithm, conditional
upon the algorithm having not stopped before time $k$. Letting $\tau_{x}$
denote the stopping time of the diagnostic when $p=p(x)$, it follows
that the probability that the algorithm never terminates is greater
than or equal to $(1-\epsilon)^{n-1}\int_{A}\nu R_{a,\epsilon}^{n-1}({\rm d}x)\mathbb{P}(\tau_{x}=\infty)>0$.

For the second part, let $A:=\{x\in\mathsf{X}:p(x)=\beta\}$ and $n=\min\{k:\nu R_{a,\epsilon}^{k-1}(A)>0\}$
as above. Then it follows that the expected number of flips of a $p$-coin
is greater than $(1-\epsilon)^{n-1}\int_{A}\nu R_{a,\epsilon}^{n-1}({\rm d}x)\mathbb{E}(\tau_{x})=\infty$.
The third part is immediate.
\end{proof}

\begin{proof}[Proof of Proposition~\ref{prop:skip_free_prob_never_positive}]
As in the proof of Proposition~\ref{prop:diagnostic} we define
$(\xi_{i})_{i\geq1}$ by $\xi_{i}:=B_{i}-\beta$ but now define $S_{n}:=\beta^{-1}\sum_{i=1}^{n}\xi_{i}=\sum_{i=1}^{n}\left(mB_{i}-1\right)$.
It is clear that $(S_{n})_{n\geq1}$ is a Markov chain on the integers
that is skip-free to the left and so we can apply \citet[Corollary~1]{brown2010some}
to obtain
\[
\mathbb{P}(\tau<\infty)=1-\frac{1-mp}{1-p}=\frac{p(m-1)}{1-p}.
\]

\end{proof}

\begin{proof}[Proof of Proposition~\ref{prop:sensitivity}]
Both kernels are uniformly ergodic with $\left\Vert \Pi^{n}(x,\cdot)-\pi\right\Vert _{{\rm TV}}\leq(1-\underbar{\ensuremath{\epsilon}})^{n}$
and $\left\Vert \tilde{\Pi}^{n}(x,\cdot)-\tilde{\pi}\right\Vert _{{\rm TV}}\leq(1-\epsilon)^{n}$.
We have, following \citet[Theorem~3.2]{mitrophanov2005sensitivity},
\[
\left\Vert \tilde{\Pi}^{n}(x,\cdot)-\Pi^{n}(x,\cdot)\right\Vert _{{\rm TV}}\leq\frac{1}{\epsilon}\sup_{x\in\mathsf{X}}\left\Vert \tilde{\Pi}(x,\cdot)-\Pi(x,\cdot)\right\Vert _{{\rm TV}}.
\]
It remains to bound $\left\Vert \tilde{\Pi}(x,\cdot)-\Pi(x,\cdot)\right\Vert _{{\rm TV}}=\sup_{A\in\mathcal{B}(\mathsf{X})}\left|\Pi(x,A)-\tilde{\Pi}(x,A)\right|$.
Letting $a\wedge b$ denote $\min\{a,b\}$ and $A\in\mathcal{B}(\mathsf{X})$,
we have
\begin{eqnarray*}
\left|\Pi(x,A)-\tilde{\Pi}(x,A)\right| & = & \left|\Pi(x,A)-\epsilon\delta_{a}(A)-(1-\epsilon)\frac{\Pi(x,A)-\left\{ p(x)\wedge\epsilon\right\} \delta_{a}(A)}{1-p(x)\wedge\epsilon}\right|\\
 & = & \left|\Pi(x,A)\left\{ 1-\frac{1-\epsilon}{1-p(x)\wedge\epsilon}\right\} +\delta_{a}(A)\left\{ \frac{1-\epsilon}{1-p(x)\wedge\epsilon}\left\{ p(x)\wedge\epsilon\right\} -\epsilon\right\} \right|\\
 & = & \left|\Pi(x,A)\left\{ \frac{\epsilon-p(x)\wedge\epsilon}{1-p(x)\wedge\epsilon}\right\} +\delta_{a}(A)\left\{ \frac{p(x)\wedge\epsilon-\epsilon}{1-p(x)\wedge\epsilon}\right\} \right|\\
 & = & \left|\Pi(x,A)-\delta_{a}(A)\right|\frac{\epsilon-p(x)\wedge\epsilon}{1-p(x)\wedge\epsilon}\\
 & \leq & \left|\Pi(x,A)-\delta_{a}(A)\right|\frac{\epsilon-\underbar{\ensuremath{\epsilon}}}{1-\underbar{\ensuremath{\epsilon}}},
\end{eqnarray*}
where the last inequality follows because $(a-b)/(1-b)$ is monotonically
decreasing as $b$ increases and $\underline{p}\geq\underbar{\ensuremath{\epsilon}}$.
Because $\Pi(x,\{a\})\geq\underbar{\ensuremath{\epsilon}}$ we have
$\left|\Pi(x,A)-\delta_{a}(A)\right|\leq1-\underbar{\ensuremath{\epsilon}}$
and therefore $\left\Vert \tilde{\Pi}(x,\cdot)-\Pi(x,\cdot)\right\Vert _{{\rm TV}}\leq\epsilon-\underbar{\ensuremath{\epsilon}}$.
Finally,
\[
\left\Vert \tilde{\Pi}^{n}(x,\cdot)-\Pi^{n}(x,\cdot)\right\Vert _{{\rm TV}}\leq\frac{1}{\epsilon}\sup_{x\in\mathsf{X}}\left\Vert \tilde{\Pi}(x,\cdot)-\Pi(x,\cdot)\right\Vert \leq1-\underbar{\ensuremath{\epsilon}}/\epsilon
\]
and we conclude.
\end{proof}

\begin{proof}[Proof of Proposition~\ref{prop:finite_expected_time_general_alg}]
The proof essentially follows the same argument as the corresponding
negative results of \citet{asmussen1992stationarity} and \citet{blanchet2007exact}.
Let $T(n)$ denote the computational time required to simulate from
$\eta_{n}^{(\nu,s)}$. Then the assumptions provide that there exists
$C<\infty$ such that the expected computational time of this simulation
is
\[
\sum_{n\geq1}T(n)\mathbb{P}_{\nu,s}(\tau_{\nu,s}\geq n)/\mathbb{E}_{\nu,s}(\tau_{\nu,s})\leq\sum_{n=1}^{\infty}Cn^{d}n^{-b},
\]
and the r.h.s.\ is finite whenever $d<b-1$.
\end{proof}

\begin{proof}[Proof of Proposition~\ref{prop:reasonablechain-boundedvariancereturn}]
From \citet[Theorem~1.1]{bednorz2008regeneration} we have that a
$\sqrt{n}$-CLT holds for a function $g$ if and only if $\mathbb{E}_{\alpha,p}\left[\left(\sum_{k=1}^{\tau_{\alpha,p}}\bar{g}(X_{k})\right)^{2}\right]<\infty$,
where $\bar{g}:=g-\pi(g)$. We define $\left\Vert g\right\Vert :=\sup_{x\in\mathsf{X}}\left|g(x)\right|$.

$\left(\Longleftarrow\right)$If $\mathbb{E}_{\alpha,p}\left(\tau_{\alpha,p}^{2}\right)<\infty$
and $\left\Vert g\right\Vert <\infty$, we have $\left\Vert \bar{g}\right\Vert \leq2\left\Vert g\right\Vert $
and so
\[
\mathbb{E}_{\alpha,p}\left[\left(\sum_{k=1}^{\tau_{\alpha,p}}\bar{g}(X_{k})\right)^{2}\right]\leq\mathbb{E}_{\alpha,p}\left[\left(\sum_{k=1}^{\tau_{\alpha,p}}2\left\Vert g\right\Vert \right)^{2}\right]=4\left\Vert g\right\Vert ^{2}\mathbb{E}_{\alpha,p}\left(\tau_{\alpha,p}^{2}\right)<\infty.
\]

$\left(\Longrightarrow\right)$ If $\mathbb{E}_{\alpha,p}\left[\left(\sum_{k=1}^{\tau_{\alpha,p}}\bar{g}(X_{k})\right)^{2}\right]<\infty$
for all $g$ such that $\left\Vert g\right\Vert <\infty$ then clearly
this holds for $g(x)=\mathbb{I}\left\{ x\in\alpha\right\} -\pi(\alpha)=\bar{g}(x)$.
Then
\begin{eqnarray*}
\mathbb{E}_{\alpha,p}\left[\left(\sum_{k=1}^{\tau_{\alpha,p}}\bar{g}(X_{k})\right)^{2}\right] & = & \mathbb{E}_{\alpha,p}\left[\left(1-\tau_{\alpha,p}\pi(\alpha)\right)^{2}\right]\\
 & = & \pi(\alpha)^{2}\mathbb{E}_{\alpha,p}\left(\tau_{\alpha,p}^{2}\right)-1
\end{eqnarray*}
as $\mathbb{E}_{\alpha,p}\left(\tau_{\alpha,p}\right)=1/\pi(\alpha)$
from Kac's theorem so it follows that $\mathbb{E}_{\alpha,p}\left(\tau_{\alpha,p}^{2}\right)<\infty$.
\end{proof}

\section{Lemmas for $F$ and $\check{F}$\label{sec:Expressions-for}}

We adopt here the notation of \citet{DelMoral2004}. We define for
each $p\in\{2,\ldots,n\}$, the non-negative kernel $Q_{p}(x_{p-1},{\rm d}x_{p}):=G_{p-1}(x_{p-1})M_{p}(x_{p-1},{\rm d}x_{p})$.
We can then define for any $p\in\{2,\ldots,n\}$ and $k\in\{1,\ldots,n-p+1\}$
\[
Q_{p,p+k}(A)(z_{p}):=\int_{\mathsf{Z}^{k}}\mathbb{I}(z_{p+k}\in A)\prod_{q=p+1}^{p+k}Q_{q}(z_{q-1},{\rm d}z_{q}),\quad A\in\mathcal{B}(\mathsf{Z}).
\]
This allows us to express, e.g., $\gamma_{n}(1)=\mu Q_{2,n+1}(1)$.
We define for each $p\in\{2,\ldots,n\}$ 
\[
\eta_{p}(A):=\frac{\mu Q_{2,p}(A)}{\mu Q_{2,p}(1)},\quad A\in\mathcal{B}(\mathsf{Z}),
\]
and the relationship between $\eta_{p}$ and $\pi_{p-1}$ is that
$\eta_{p}(A)=\pi_{p-1}(f)$ where $f(z_{1},\ldots,z_{p-1})=M_{p}(z_{p-1},A)$.
It follows that $\eta_{p}(G_{p})=\gamma_{p}(1)/\gamma_{p-1}(1)$ and
that $\eta_{p}Q_{p,p+k}(1)=\gamma_{p+k-1}(1)/\gamma_{p-1}(1)$.

For a Markov kernel $P$ evolving on $E$, we write $P(f)(x):=\int_{\mathsf{Z}}f(z)P(x,{\rm d}z)$.
Assumption (A1) of \citet{andrieu2013uniform} is the existence of
a constant $A<\infty$ such that 
\begin{equation}
\sup_{z\in\mathsf{Z}}\frac{Q_{p,p+k}(1)(z)}{\eta_{p}Q_{p,p+k}(1)}\leq A,\qquad p\in\{2,\ldots,n\},k\in\{1,\ldots,n-p+1\},\label{eq:condition_icsmcA}
\end{equation}
and this is in fact exactly $F$ defined in Section~\ref{sub:Iterated-conditional-SMC}.

Let $S_{p}:=\{z\in\mathsf{Z}:G_{p}(z)>0\}$ and $\bar{G}_{p}:=\sup_{z\in\mathsf{Z}}G_{p}(z)$
for each $p\in\{1,\ldots,n\}$, with $S_{0}:=\mathsf{Z}$. Let $M_{n+1}(z,A):=\chi(A)$
independent of $z$ and 
\[
M_{p+1,p+m}(z_{p},A):=\int_{\mathsf{Z}^{m}}\mathbb{I}(z_{p+m}\in A)\prod_{q=p+1}^{p+m}M_{q}(z_{q-1},{\rm d}z_{q}),\quad A\in\mathcal{B}(\mathsf{Z}).
\]

\begin{lem}
\label{lemma:bound_on_A_alternative}Let $A<\infty$, and assume that
for all $p\in\{1,\ldots,n\}$, there exists an $m\in\{1,\ldots,n-p+1\}$
such that 
\[
\frac{\prod_{q=p}^{p+m}\bar{G}_{q}}{\prod_{q=p}^{p+m}\eta_{q}(G_{q})}\cdot\sup_{z_{p}\in S_{p},z_{p+m}\in S_{p+m}}\frac{M_{p+1,p+m}(z_{p},{\rm d}z_{p+m})}{\eta_{p+m}({\rm d}z_{p+m})}\leq A,
\]
then (\ref{eq:condition_icsmcA}) holds with this choice of $A$.\end{lem}
\begin{proof}
If $k\leq m$, then the result holds trivially. For any $k\in\{m+1,\ldots,n-p+1\}$,
\begin{eqnarray*}
Q_{p,p+k}(1)(z_{p}) & = & \int_{\mathsf{Z}^{k}}\prod_{q=p+1}^{p+k}Q_{q}(z_{q-1},{\rm d}z_{q})\\
 & = & \int_{\mathsf{Z}}Q_{p,p+m}(z_{p},{\rm d}z_{p+m})Q_{p+m,p+k}(1)(z_{p+m})\\
 & = & \int_{\mathsf{Z}}\frac{Q_{p,p+m}(z_{p},{\rm d}z_{p+m})}{\left[\prod_{q=p}^{p+m-1}\eta_{q}(G_{q})\right]\eta_{p+m}({\rm d}z_{p+m})}\left[\prod_{q=p}^{p+m-1}\eta_{q}(G_{q})\right]\eta_{p+m}({\rm d}z_{p+m})Q_{p+m,p+k}(1)(z_{p+m})\\
 & \leq & \left[\frac{\prod_{q=p}^{p+m-1}\bar{G}_{q}}{\prod_{q=p}^{p+m-1}\eta_{q}(G_{q})}\sup_{z_{p}\in S_{p},z_{p+m}\in S_{p+m}}\frac{M_{p+1,p+m}(z_{p},{\rm d}z_{p+m})}{\eta_{p+m}({\rm d}z_{p+m})}\right]\eta_{p}Q_{p,p+k}(1).
\end{eqnarray*}
\end{proof}
\begin{cor}
\label{cor:bound_on_A}If
\[
\sup_{z_{p-1}\in S_{p-1},z_{p}\in S_{p}}\frac{G_{p}(z_{p})}{M_{p}(G_{p})(z_{p-1})}\cdot\sup_{z_{p},z'_{p}\in S_{p},z_{p+1}\in S_{p+1}}\frac{M_{p+1}(z_{p},{\rm d}z_{p+1})}{M_{p+1}(z'_{p},{\rm d}z_{p+1})}\leq A,
\]
then (\ref{eq:condition_icsmcA}) holds with this choice of $A$.\end{cor}
\begin{proof}
The result follows from taking $m=1$ in Lemma~\ref{lemma:bound_on_A_alternative}
and noting that\linebreak{}
 $\eta_{p+1}({\rm d}z_{p+1})\geq\inf_{z_{p}\in S_{p}}M_{p+1}(z_{p},{\rm d}z_{p+1})$
and $\eta_{p}(G_{p})\geq\inf_{z_{p-1}\in S_{p-1}}M_{p}(G_{p})(z_{p-1})$.
\end{proof}
When considering the i-cSMC kernel associated with the extended Feynman--Kac
model of Section~\ref{sub:Atomic-extensions-of}, it is possible
that the choice of $A$ satisfying (\ref{eq:condition_icsmcA}) for
the original model does not simultaneously satisfy the associated
condition for the extended model. Indeed, one now requires a constant
$\check{A}<\infty$ such that 
\begin{equation}
\sup_{z\in\mathsf{Z}}\frac{\check{Q}_{p,p+k}(1)(z)}{\check{\eta}_{p}\check{Q}_{p,p+k}(1)}\leq\check{A},\qquad p\in\{2,\ldots,n\},k\in\{1,\ldots,n-p+1\},\label{eq:condition_icsmcA_check}
\end{equation}
where it can be seen that $\check{A}$ is equal to $\check{F}$ defined
in Section~\ref{sub:Atomic-extensions-of}.
\begin{lem}
\label{lem:Acheck}Let $E\geq1$ be such that for each $p\in\{1,\ldots,n\}$,
$E^{-1}\leq\psi_{p}/\eta_{p}(G_{p})\leq E$. Then (\ref{eq:condition_icsmcA_check})
holds with $\check{A}=AE^{n}$, where $A$ satisfies (\ref{eq:condition_icsmcA}).\end{lem}
\begin{proof}
We can write
\begin{eqnarray*}
\check{\eta}_{p}\check{Q}_{p,p+k}(1) & = & \check{\eta}_{p}(\{a\})\prod_{q=p}^{p+k-1}\psi_{q}+\check{\eta}_{p}(\mathsf{Z})\eta_{p}Q_{p,p+k}(1)\\
 & = & \eta_{p}Q_{p,p+k}(1)\left(\check{\eta}_{p}(\{a\})\frac{\prod_{q=p}^{p+k-1}\psi_{q}}{\eta_{p}Q_{p,p+k}(1)}+\check{\eta}_{p}(\mathsf{Z})\right).
\end{eqnarray*}
Now we have for $z=a$ 
\[
\frac{\check{Q}_{p,p+k}(1)(a)}{\check{\eta}_{p}\check{Q}_{p,p+k}(1)}=\frac{\prod_{q=p}^{p+k-1}\psi_{q}}{\check{\eta}_{p}(\{a\})\prod_{q=p}^{p+k-1}\psi_{q}+\check{\eta}_{p}(\mathsf{Z})\eta_{p}Q_{p,p+k}(1)}\leq\max\left\{ 1,\frac{\prod_{q=p}^{p+k-1}\psi_{q}}{\eta_{p}Q_{p,p+k}(1)}\right\} ,
\]
whereas for $z\in\mathsf{Z}$ 
\begin{eqnarray*}
\frac{\check{Q}_{p,p+k}(1)(z)}{\check{\eta}_{p}\check{Q}_{p,p+k}(1)} & = & \frac{Q_{p,p+k}(1)(z)}{\eta_{p}Q_{p,p+k}(1)}\left(\check{\eta}_{p}(\{a\})\frac{\prod_{q=p}^{p+k-1}\psi_{q}}{\eta_{p}Q_{p,p+k}(1)}+\check{\eta}_{p}(\mathsf{Z})\right)^{-1}\\
 & \leq & A\max\left\{ 1,\frac{\eta_{p}Q_{p,p+k}(1)}{\prod_{q=p}^{p+k-1}\psi_{q}}\right\} .
\end{eqnarray*}
Since $\max\left\{ 1,\frac{\eta_{p}Q_{p,p+k}(1)}{\prod_{q=p}^{p+k-1}\psi_{q}}\right\} \leq E^{n}$
and $\max\left\{ 1,\frac{\prod_{q=p}^{p+k-1}\psi_{q}}{\eta_{p}Q_{p,p+k}(1)}\right\} \leq E^{n}$
the result is immediate.
\end{proof}

\section{An alternative solution to Problem~\ref{prob:bf2} via a generic
solution to the sign problem\label{sec:An-alternative-solution}}

We have $(1-p)/(1-\epsilon)=1-(p-\epsilon)/(1-\epsilon)$ so simulation
of a $(1-p)/(1-\epsilon)$-coin is possible whenever one can simulate
a $(p-\epsilon)/(1-\epsilon)$-coin. This in turn is feasible when
one can simulate a random variable $V$ satisfying $\mathsf{E}(V)=p$
and $\mathsf{P}(\epsilon\leq V\leq1)=1$, since the random variable
$\mathbb{I}(U<\frac{V-\epsilon}{1-\epsilon})$ is a $(p-\epsilon)/(1-\epsilon)$-coin
flip, where $U$ is an independent uniform random variable on $[0,1]$.
We now show that unbiased constrained estimation is possible in a
more general sense, and a solution to Problem~\ref{prob:bf2} follows
naturally whenever $\epsilon\leq1/2$.

The first step is a generic solution to what is often referred to
as a ``sign problem''. In particular, given a probability measure
$\mu:\mathcal{B}(\mathsf{X})\rightarrow[0,1]$ whose distribution
one can obtain any number of i.i.d.\ samples from, and a function
$\varphi:\mathsf{X}\rightarrow\mathbb{R}$ such that $\mu(\varphi)>0$,
one wishes to simulate a random variable with expectation $\mu(\varphi)$
that is almost surely non-negative. The non-existence of a general
procedure for solving this problem without any assumptions on $\mu$
and $\varphi$ has been shown in \citet{jacob2013non} following other
non-existence results such as \citet{keane1994bernoulli} and \citet{henderson2003nonexistence}.
We provide a positive result in the case where $\varphi$ is a bounded
function and $\mu(\varphi)\geq\delta>0$ for a known constant $\delta$.
Our scheme relies on the probability measure
\[
\mu_{\left|\varphi\right|}({\rm d}x):=\frac{\mu({\rm d}x)\left|\varphi(x)\right|}{\mu(\left|\varphi\right|)},
\]
which can be sampled from by rejection whenever $\left\Vert \varphi\right\Vert :=\sup_{x\in\mathsf{X}}\left|\varphi(x)\right|<\infty$.
A general algorithm for solving the sign problem is provided in Algorithm~\ref{alg:generic_sign_problem},
whose validity is established in Proposition~\ref{prop:general_bounded_signed_lowerbound}.

\begin{algorithm}

\protect\caption{Simulate an unbiased and almost surely positive estimate of $\mu(\varphi)\geq\delta>0$\label{alg:generic_sign_problem}}

\begin{enumerate}
\item Sample $\xi\sim\mu$.
\item Flip a $2q$-coin $Y$, where $2q\leq1-\delta/\left\Vert \varphi\right\Vert $,
using a Bernoulli factory using flips of a $q$-coin, each such flip
being simulated by sampling $\zeta\sim\mu_{\left|\varphi\right|}$
and outputting $\mathbb{I}({\rm sign}(\varphi(\zeta))=-1)$.
\item Output $\left|\varphi(\xi)\right|(1-Y)$.\end{enumerate}
\end{algorithm}

\begin{prop}
\label{prop:general_bounded_signed_lowerbound}Let $\varphi$ satisfy
$\left\Vert \varphi\right\Vert <\infty$ and $\mu(\varphi)\geq\delta>0$.
Then the output of Algorithm~\ref{alg:generic_sign_problem} is a
random variable $W$ such that $\mathsf{E}(W)=\mu(\varphi)$ and $\mathsf{P}(0\leq W\leq\left\Vert \varphi\right\Vert )=1$.\end{prop}
\begin{proof}
We denote the output of Algorithm~\ref{alg:generic_sign_problem}
as $W=\left|\varphi(\xi)\right|(1-Y)$. Since $\mu(\varphi)>0$ implies
that $\mu(\left|\varphi\right|)>0$, we can express $\mu(\varphi)$
as
\[
\mu(\varphi)=\mu(\left|\varphi\right|)\frac{\mu(\varphi)}{\mu(\left|\varphi\right|)}=\mu(\left|\varphi\right|)\int_{\mathsf{X}}{\rm sign}(\varphi(x))\mu_{\left|\varphi\right|}({\rm d}x).
\]
The random variable $W$ is the product of independent unbiased estimates
of $\mu(\left|\varphi\right|)$ and $\int_{\mathsf{X}}{\rm sign}(\varphi(x))\mu_{\left|\varphi\right|}({\rm d}x)$
respectively. That $\mathsf{E}\left|\varphi(\xi)\right|=\mu(\left|\varphi\right|)$
is immediate, and it is clear that $\mathsf{P}(0\leq\left|\varphi(\xi)\right|\leq\left\Vert \varphi\right\Vert )=1$.
We now show that $\mathsf{E}(1-Y)=\int_{\mathsf{X}}{\rm sign}(\varphi(x))\mu_{\left|\varphi\right|}({\rm d}x)$.
By construction,
\[
\int_{\mathsf{X}}{\rm sign}(\varphi(x))\mu_{\left|\varphi\right|}({\rm d}x)=\mathsf{E}\left[{\rm sign}(\varphi(\zeta))\right]=\mu(\varphi)/\mu(\left|\varphi\right|)>0,
\]
since $\mu(\varphi)\geq\delta>0$. Furthermore, ${\rm sign}(\varphi(\zeta))$
is almost surely valued in $\{-1,+1\}$ since $\zeta\sim\mu_{\left|\varphi\right|}$,
and so
\[
\mathsf{E}\left[{\rm sign}(\varphi(\zeta))\right]=p-(1-p)=2p-1,
\]
where $p:=\mathsf{P}({\rm sign}(\varphi(\zeta))=1)>\frac{1}{2}$.
From the bound 
\[
2p-1=\mu(\varphi)/\mu(\left|\varphi\right|)\geq\delta/\left\Vert \varphi\right\Vert ,
\]
it follows that $p\geq\frac{1}{2}+\frac{\delta}{2\left\Vert \varphi\right\Vert }$.
Since $q=\mathsf{P}({\rm sign}(\varphi(\zeta))=-1)=(1-p)$, simulating
a $(2p-1)$-coin is equivalent to simulating a $(1-2q)$-coin with
$q\leq\frac{1}{2}-\frac{\delta}{2\left\Vert \varphi\right\Vert }$,
and this is exactly what $(1-Y)$ is. It follows that $\mathsf{E}(W)=\mu(\varphi)$
and $\mathsf{P}(0\leq W\leq\left\Vert \varphi\right\Vert )=1$.\end{proof}
\begin{rem*}
It is clear that a number of variance reduction techniques can be
used to modify Algorithm~\ref{alg:generic_sign_problem}. For example,
one could simply average a number of outputs of the algorithm. In
addition, one could average different values of $\left|\varphi(\xi)\right|$
and/or the $(1-2q)$-coins.
\end{rem*}
An immediate Corollary of Proposition~\ref{prop:general_bounded_signed_lowerbound}
is that more general constrained unbiased estimation problems can
be solved.
\begin{cor}
\label{cor:constrained_unbiased}Let $\varphi$ satisfy $a\leq\varphi(x)\leq c$
and assume we can simulate random variables according to $\mu$. If
$\mu(\varphi)\geq\delta>b$ for some known $\delta$ then we can simulate
a random variable $W$ such that $\mathsf{E}(W)=\mu(\varphi)$ and
$\mathsf{P}(b\leq W\leq b+\max\{b-a,c-b\})=1$.\end{cor}
\begin{proof}
Let $\tilde{\varphi}=\varphi-b$. Then $\left\Vert \tilde{\varphi}\right\Vert =\max\{b-a,c-b\}$
and by Proposition~\ref{prop:general_bounded_signed_lowerbound}
we can simulate a random variable $\tilde{W}$ with $\mathsf{E}\tilde{W}=\mu(\tilde{\varphi})$
and $\mathsf{P}(0\leq\tilde{W}\leq\left\Vert \tilde{\varphi}\right\Vert )=1$.
We conclude by letting $W=\tilde{W}+b$.
\end{proof}
It follows from Corollary~\ref{cor:constrained_unbiased} that if
we can simulate a $p$-coin with $p\geq\delta>\epsilon$, then we
can simulate a random variable $V$ satisfying $\mathsf{E}(V)=p$
and $\mathsf{P}(\epsilon\leq V\leq1)=1$ when $\epsilon\leq1/2$,
by taking $\mu$ to be ${\rm Bernoulli}(p)$, $\varphi(x)=x$, $a=0$,
$b=\epsilon$ and $c=1$. While we do not recommend the use of Algorithm~\ref{alg:generic_sign_problem}
to obtain a solution to Problem~\ref{prob:bf2}, it may be useful
more generally. As a simple example, consider the case where $\mu$
is the uniform distribution on $[0,3\pi]$ and $\varphi=\sin$. Then
$\mu(\varphi)=2/(3\pi)>0$ and one use Algorithm~\ref{alg:generic_sign_problem}
to produce a random variable with expectation $\mu(\varphi)$ and
which is almost surely in $[0,1]$.

{\small{}\bibliographystyle{Chicago}
\bibliography{atomicsmc}
}
\end{document}